\documentclass[onecolumn]{IEEEtran}
\IEEEoverridecommandlockouts
\usepackage{cite}
\usepackage{amssymb,amsfonts}
\usepackage{algorithmic}
\usepackage{graphicx}
\usepackage{textcomp}
\usepackage[tbtags]{amsmath}
\usepackage{diagbox}
\usepackage{multirow}
\usepackage{amsthm} 
\usepackage{algorithm}
\newtheorem{theorem}{Theorem}

\theoremstyle{definition}
\usepackage{comment}
\usepackage{booktabs,color,amssymb,amsmath}

\def\BibTeX{{\rm B\kern-.05em{\sc i\kern-.025em b}\kern-.08em
    T\kern-.1667em\lower.7ex\hbox{E}\kern-.125emX}}

\newtheorem{Definition}{Definition}
\newtheorem{Example}{Example}
    
\begin{document}

\title{Reduce Transmission Delay for Cache-Aided Relay Networks \\
}

\author{
  \IEEEauthorblockN{Ke Wang\IEEEauthorrefmark{1}\IEEEauthorrefmark{2}\IEEEauthorrefmark{3}, Youlong Wu\IEEEauthorrefmark{1}, Shujie Cao\IEEEauthorrefmark{1}, Jiahui Chen\IEEEauthorrefmark{1}\IEEEauthorrefmark{2}\IEEEauthorrefmark{3}   \\}
\IEEEauthorblockA{\IEEEauthorrefmark{1}School of Information Science and Technology,
	ShanghaiTech University, Shanghai 201210, China\\}
\IEEEauthorblockA{\IEEEauthorrefmark{2} Shanghai Institute of Microsystem and Information Technology,
	Chinese Academy of Sciences\\}
\IEEEauthorblockA{\IEEEauthorrefmark{3} University of Chinese Academy of Sciences, Beijing 100049, China
  	\{wangke,wuyl1,caoshj,chenjh1\}@shanghaitech.edu.cn}
  \thanks{This paper was in part presented at the \emph{IEEE International Symposium on Information Theory (ISIT)},  Paris, France, July, 2019, and   in part submitted to ISIT,  L. A., California, US, 2020.

Shujie Cao, Youlong Wu, Jiahui Chen, Ke Wang and Haoyu Tu are with the 
School of Information Science and Technology, ShanghaiTech University, 201210 Shanghai, China. (e-mail:  \{caoshj,wuyl1,chenjh1,wangke,tuhy\}@shanghaitech.edu.cn).
}
}


\maketitle

\begin{abstract}
 
    In this paper, we consider a cache-aided relay network, where  a single server consisting of a library of $N$ files connects with $K_{1}$ relays through a shared noiseless link, and each relay  connects with $K_{2}$ users through a shared noiseless link. Each relay and user are equipped with a cache memory of $M_1$ and $M_2$ files, respectively.  We  propose a centralized  and  a decentralized coded caching scheme that  exploit the spared transmission time resource by allowing  concurrent transmission between the two layers. It is shown that both caching schemes are approximately optimal, and  greatly reduce the transmission delay compared to  the previously known caching schemes. Surprisingly, we show that when the relay's caching size is equal to a threshold that  is strictly smaller than $N$ (e.g. $M_1=0.382N$ under the decentralized setup and $(K_1-1)N/{K_1}$ under the centralized setup, when $K_1=2$), our schemes achieve the same delay as if  each relay had access to the full library. To our best knowledge, this is the first result showing that even the  caching size is strictly smaller than the library's size, increasing the caching size  is wasteful in reducing the transmission latency. 
                         
\end{abstract}


\section{Introduction}
Caching is considered as a promising technique to release the traffic load on the Internet during network peak hours. A representative approach is to take advantage of the cache memories of end nodes or other terminals to store some contents in advance. Thus, only the contents that are not cached in local caches need to be delivered, resulting in a reduction on communication load. The whole procedure in the caching system is divided into two phases: the placement phase, where each user prefetches some contents to fill its local cache, and the delivery phase, where users inform their demands to the server and the server delivers the contents needed by the users according to the information cached by users. To further reduce the traffic load and improve the transmission efficiency, Maddah-Ali and Niesen proposed \emph{coded caching }  which obtains a global caching gain by creating multicasting opportunities for multiple users in \cite{Centralized,Decentralized}. 


Caching problem on different relay networks was considered in  \cite{Karamchandani'16,Zewail'17,Tao'ISIT18,Sengupta'17}, in which a server  communicates with multiple users with the help of multiple relays. In particular, the work in \cite{Karamchandani'16} considered a noiseless network where one server connects with multiple relays with each relay serving a distinct set of users. For this network, the authors proposed a hierarchical coded caching (HCC) scheme which achieves the optimal communication rates  within a constant multiplicative and additive gap. In \cite{Zewail'17} it investigated a network where the server connects with each relay via individual link, and each user is connected to a distinct set of relay nodes. A more general network where each relays connects with all users through wireless channel is considered in \cite{Tao'ISIT18,Sengupta'17}. 

Some other aspects of coded caching have been investigated in the literature. In  \cite{fileuser}, it studied the case of different users requesting the same file, where more users than files in centralized coded caching. In \cite{physical}, a coded caching scheme achieving  both spatial multiplexing and buffer gain by using coded delivery and zero-forcing was proposed. Coded caching with multiple transmit antennas was studied \cite{addtransmitters}. A special structure called \textit{placement delivery array (PDA)} was proposed to describe placement and delivery phase in coded caching schemes in a simpler manner with reduced subpacketization \cite{PDA1}. Hypergraphs and bipartite graphs are employed to describe coded caching and present the schemes with subpacketization
subexponential in $K$ \cite{hypergraph, PDA2}. In \cite{withoutfilesplit}, keeping the files intact during coded caching was proposed for simpler implementation, also reducing the delivery rate. Information security and private information retrieval are also introduced to coded caching to expand the radiation field of the technique \cite{private1, private2, private3}. Other work on coded caching  include, e.g. cache system with  heterogeneous problem settings \cite{DffDemand,HeterogeneousCacheSizes,cHeterogeneous'17,cHeterogeneous'19,HeterogeneousOptimization,cachecapacity}, cache-aided noiseless multi-server network \cite{multi-server}, cache-aided device-to-device network \cite{D2D}, cache-aided interference management \cite{intermanage'17}, \cite{b11}, coded caching with distinct sizes of fies \cite{filesize}, coded-caching with random demands \cite{randomdemand}, coded caching based on combinatorial designs \cite{combinationdesign}, combining with distributed computing for a tradeoff between computation and communication \cite{distributedcomputing}, caching in combination networks \cite{combinationnetwork}, etc.

In this paper, we revisit the  relay  network considered in \cite{Karamchandani'16}. More specifically, we study a  two-layer network where   a single server consisting of a library of $N$ files connects with multiple relays,   each  equipped with a cache memory of $M_1$ files, via a shared noiseless link, and each relay connects with a distinct set of users, each equipped with a cache memory of  $M_2$ files, via a shared noiseless link.  Since the server and relays operate in two separate layers, we assume that the relays can send signal during the server's transmission.
The main  contributions of this paper are summarized as follows.

\begin{itemize}{ 
	\item We   show that a simple pipelined forward scheme  outperforms  HCC scheme by exploiting  the opportunity of  concurrent transmission at the server and relays.   As we will show later,  an  intrinsic property of scheme HCC, excludes it from utilizing this opportunity. We combine   the original HCC scheme with pipeline-forward and give an approximately optimal choice of system splitting parameters. 
	
	\item We propose novel centralized and  decentralized caching schemes that fully exploit the spared  time resource by  letting the server and the relays seamlessly send data, i.e., relays are allowed to send and receive data simultaneously.  The schemes are approximately optimal, and  greatly reduce the transmission delay compared to  the HCC schemes.   Instead of  regarding these two layers separately in HCC scheme,  we jointly design the file placement and delivery for them, which makes our work  not an easy extension of Maddah-Ali and Niesen's schemes \cite{Centralized, Decentralized} to the two-layer network.
	

	\item  Surprisingly, we show that when each relay's caching size is equal to a threshold that is strictly smaller than $N$, increasing caching size at relay nodes will not reduce  the transmission latency. More specifically, under the the centralized caching placement as \cite{Centralized}, if $M_1$ equals to the threshold $(K_1-1)N/{K_1}$,  our scheme can achieve  the optimal transmission delay as if all relays had access to the full library; under the decentralized caching placement as \cite{Decentralized}, if $K_1=2,M_1=0.38N$, then  our scheme can again achieve  the optimal transmission delay as if all relays had access to the full library.

	}
\end{itemize}

The rest of the paper is organized as follows. Section \ref{eq:Model} introduces the system model considered in this paper. Section \ref{Sec_Preliminary} reviews the related work. Section \ref{Sec_Motivation} gives  motivations and some simple examples of our new schemes. Section \ref{sec_results} presents our main results.  The proposed centralized and decentralized schemes are described in Section \ref{sec_proof_centralized} and Section \ref{sec_proof_decentralized}, respectively. Section \ref{sec_conclusion} concludes the paper.

\section{Problem Definition} \label{eq:Model}
Consider a two-layer  delivery network in Fig. \ref{fig_model}, which includes a single server, $K_1$ relays and $K_1K_2$ users. 
 The server has a library of $N$ independent files $W_1,\ldots,W_N$. Each $W_n$, $n=1,\ldots,N$, is  is uniformly distributed over \[[2^F]\triangleq \{1,\ldots,2^F\},\] for some positive integer $F$. Every relay node  has   a cache memory of size  $M_1F$ bits, $M_1\in[0,N]$, and is connected to the server  through a noiseless shared link. Meanwhile, each relay  connects with $K_2$ users, each equipped with a  cache memory of size   $M_2F$ bits,  for $M_2\in[0,N]$, through a noiseless shared link. 
Let the $j$-th user attached to relay $i$ be $u^i_j$, for $i\in[K_1]$, and $j\in[K_2]$ and define   
\begin{IEEEeqnarray}{rCl}
&& \mathcal{U}_{i}\triangleq  \{u^i_1,\ldots, u^i_{K_2}\},\\
 && \mathcal{U}\triangleq  \mathcal{U}_{1}\cup\cdots\cup\mathcal{U}_{K_1}
\end{IEEEeqnarray}
where  $\mathcal{U}_i$  and $\mathcal{U}$ denote the set of the users' indices with respect to relay $i$ and the set of \emph{all} users' indices, respectively.

\begin{figure}
           \centering
	\includegraphics[width=0.5\textwidth]{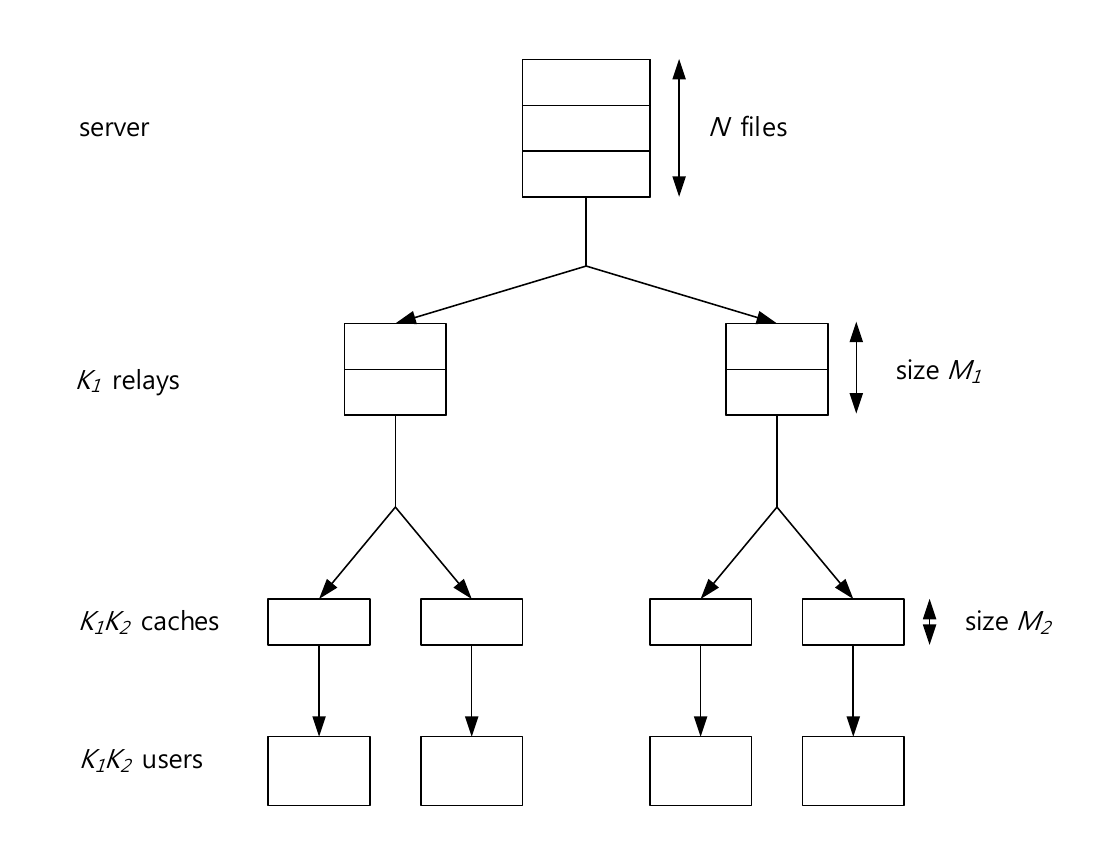}
	\caption{A two-layer network. A server containing a library of $N$ files connects with $K_1$ relays, and each relay  connects with $K_2$ users, both through shared noiseless links. The cache sizes of each relay and user are $M_1$ and $M_2$, respectively. In this figure, $N=4$, $K_1=K_2=2$, $M_1=2$ and $M_2=1$.}
	\label{fig_model}
\end{figure}

Each user  requests one of the $N$ files from the library. We denote the demand of  user $u^i_j$ by 
$d_j^{~\!\!i}\in[N]$,
and the corresponding file  by  $W_{d_j^{~\!\!i}}$. Let 
\[\mathbf{d}\triangleq (d_j^{~\!\!i},\ldots,d^{K_1}_{K_2})\] denote the users' request vector.

The system operates in two phases: a placement phase and a delivery phase. 
In the placement phase, all relays and users have access to the entire library $W_1,\ldots,W_N$ and fill the content to their caches. More specifically,  each relay $i$ maps the files $W_1,\ldots,W_N$ to the cache content: 
\begin{IEEEeqnarray}{rCl}
&&Z_{i}\triangleq \phi_{i}(W_1,\ldots,W_N),
\end{IEEEeqnarray}
and user $u^i_j$ maps  $W_1,\ldots,W_N$ to the cache content: 
\begin{IEEEeqnarray}{rCl}
Z^i_j\triangleq \phi^i_j(W_1,\ldots,W_N),
\end{IEEEeqnarray}
for some caching functions
 \begin{subequations}\label{eq:caching}
\begin{IEEEeqnarray}{rCl}
&&\phi_{i}: [2^F]^N\rightarrow [\lfloor2^{FM_1}\rfloor],~ \phi^i_j:[2^F]^N \rightarrow [\lfloor2^{FM_2}\rfloor].
\end{IEEEeqnarray}
\end{subequations}

 In the delivery phase, the server and relays are informed with the users' requests $\bf{d}$, and  send symbols to the network such that  all the users' requests are satisfied. More specifically, the server produces symbol \[X_{\bf{d}} \triangleq  f_{\bf{d}}(W_1,\ldots,W_N),\]
  and relay $i$ perfectly receives the symbols sent by the server, and produces symbol \[X_{i,{\bf{d}}}\triangleq  f_{i,\bf{d}}(Z_i,X),\]
 for some encoding functions
 \begin{subequations}\label{eq:encoding}
\begin{IEEEeqnarray}{rCl}
&&f_{\bf{d}}:[2^F]^N\rightarrow [\lfloor2^{FR_1}\rfloor],\\
&&f_{i,\bf{d}}:[\lfloor2^{FM_1}\rfloor]\times [\lfloor2^{FR_1}\rfloor] \rightarrow  [\lfloor2^{FR_2}\rfloor]
\end{IEEEeqnarray}
\end{subequations}
 where $R_1$ and $R_2$ denote the rate   transmitted in the first layer and second layer, respectively.  
  
Each  user $u^i_j\in\mathcal{U}$ perfectly observes the symbol sent by relay $i$, and decodes its desired message as
  \[\hat{W}_{d_j^{~\!\!i}}=\psi^i_{j,\bf{d}}(X_i,Z^i_j),\]
 for some decoding function
 \begin{IEEEeqnarray}{rCl}\label{eq:decoding}
 \psi^i_{j,\bf{d}}:   [\lfloor2^{FR_2}\rfloor]\times   [\lfloor2^{FM_2}\rfloor]\rightarrow [2^F].
\end{IEEEeqnarray}
 We define the worst-case probability of error as
\begin{IEEEeqnarray}{rCl}
P_e \triangleq  \max_{{\bf{d}}\in [2^F]^N} \max_{i\in[K_1], j\in[K_2]} \text{Pr} \left(\hat{W}_{d_j^{~\!\!i}} \neq {W}_{d_j^{~\!\!i}}\right).
\end{IEEEeqnarray}

A caching scheme $(M_1,M_2,R_1,R_2)$ consists of caching functions \eqref{eq:caching}, encoding functions \eqref{eq:encoding} and decoding functions \eqref{eq:decoding}. We say that a rate region $(M_1,M_2,R_1,R_2)$  is \emph{achievable} if for every $\epsilon>0$ and every large enough file size $F$, there exists a caching scheme such that $P_e$    is less than $\epsilon$.

%
%
%
%
\begin{Definition}
	Consider a $K$-node cache-aided network and an achievable caching scheme that delivers data in $L\in \mathbb{Z}^+$ slots. Denote the transmission rate sent by node $k\in[K]$ in slot $\ell\in[L]$ as $R_{k,\ell}$.  The transmission delay of slot   $\ell$ is defined as the maximum transmission rate of all nodes in this slot, i.e., \begin{IEEEeqnarray}{rCl}
		{T}_\ell \triangleq \max_{k\in[K]}R_{k,l}.
	\end{IEEEeqnarray}
	The transmission delay of the system is defined as the total transmission delay in $L$ slots, i.e.,
	$
	T \triangleq  \sum_{\ell\in[L]}  {T}_\ell. 
	$ The \emph{optimal}  transmission delay is minimum transmission delay of all achievable caching scheme, i.e., 
	\begin{IEEEeqnarray}{rCl}
		T^*\triangleq \inf\{ T\}.
	\end{IEEEeqnarray}
\end{Definition}
\section{Preliminary: Hierarchical Coded Caching}\label{Sec_Preliminary}
Define $[x]^+\triangleq  \mathsf{max}\{x,0\}$, and 
\begin{IEEEeqnarray}{rCl}
r_\mathsf{d}\Big(\frac{M}{N},K\Big)\triangleq  \left[K\Big(1-\frac{M}{N}\Big)\frac{N}{KM}\Big(1-\Big(1-\frac{M}{N}\Big)^K\Big)\right]^+.\nonumber\\
\end{IEEEeqnarray}


In \cite{Karamchandani'16} the authors considered the  two-layer network as described in Section \ref{eq:Model}. They propose three hierarchical caching schemes, HCC-I, HCC-II and  HCC-III, based on the single-layer  decentralized caching scheme \cite{Decentralized}. We recall these three schemes for future reference and comparison as follows.

\begin{itemize}
\item Scheme HCC-I:  The main idea is that  each relay is considered as a ``tycoon user" that wishes to cover all the files requested by its attached users,  the server firstly  delivers  the requested files to all relays, and after the relays decoding all their required files, they concurrently send the requested files to their attached users. 
It is easy to obtain that to ensure all relays perfectly know $\{W_{d^{~\!\!i}_1},\ldots,W_{d^{~\!\!i}_{K_2}}\}_{i=1}^{K_1}$, the rate satisfies
\begin{subequations}\label{eq_HccARate}
\begin{equation}
R_{\mathsf{Hcc},1}^{\text{\tiny I}}\triangleq K_2\cdot r_\mathsf{d}\left({M_1}/{N},K_1\right). 
\end{equation}

In the second subphase, after the relays successfully decode all the requested files of  their attached users through the first subphase,  they parallelly  deliver   the requested files to their attached users using the single-layer decentralized  caching scheme, which leads to an achievable rate
  \begin{equation}
R^{\text{\tiny I}}_{\mathsf{Hcc},2} \triangleq r_\mathsf{d}\left({M_2}/{N},K_2\right). 
\end{equation}
\end{subequations}

\item  Scheme HCC-II: In scheme HCC-II the caching memories of the relays are completely ignored, and the relays  forward relevant parts of the server transmissions to the corresponding users. 
The rate of the first layer is
\begin{subequations}\label{eq_HccBRate}
\begin{equation}
R^{\text{\tiny II}}_{\mathsf{Hcc},1}\triangleq r_\mathsf{d}\left({M_2}/{N},K_1K_2\right),
\end{equation}
and the rate of the second layer is
\begin{equation}
R_{\mathsf{Hcc},2}^{\text{\tiny II}}\triangleq r_\mathsf{d}\left({M_2}/{N},K_2\right). 
\end{equation}
\end{subequations}

\item  Scheme  HCC-III: Informally,  HCC-III is a mixture of scheme HCC-I and scheme HCC-II. The system is divided into two subsystems with two fixed parameters $\alpha,\beta\in [0,1]$. The first subsystem includes the entire cache memory of each relay, an $\alpha$ fraction of each file in the library and a $\beta$ fraction of cache memory for each user, and the second subsystem includes the remaining $1-\alpha$ fraction of each file in the server and a $1-\beta$ fraction of each user's cache memory. Obviously, scheme HCC-I and HCC-II can be implemented in the first subsystem and  second subsystem, respectively. Thus, we have the rate of the first layer: 
\begin{subequations}\label{eq_HccCRate}
\begin{equation}
\begin{aligned}
R^{\text{\tiny III}}_{\mathsf{Hcc},1}\triangleq  \ &\alpha K_2 r_\mathsf{d}\Big(\frac{M_1}{\alpha N},K_1\Big)+(1-\alpha)\\&\cdot r_\mathsf{d}\Big(\frac{(1-\beta)M_2}{(1-\alpha)N},K_1K_2\Big),
\end{aligned}
\end{equation}
and the rate of the second layer:
\begin{equation}\label{eq_rateHRate2}
\begin{aligned}
R^{\text{\tiny III}}_{\mathsf{Hcc},2}\triangleq \ & \alpha r_\mathsf{d}\Big(\frac{\beta M_2}{\alpha N},K_1\Big)+(1-\alpha)\\&\cdot r_\mathsf{d}\Big(\frac{(1-\beta)M_2}{(1-\alpha)N},K_2\Big).
\end{aligned}
\end{equation}
The approximately optimal $\alpha$ and $\beta$ is given as
\begin{equation}\label{eqHCCab} 
	(\alpha,\beta)\!=\!\left\{
	\begin{aligned}
	&\!\Big(\frac{M_1}{N},\frac{M_1}{N}\Big), ~M_1\!+\!M_2K_2\!\ge \!N , 0\!\le \!M_1\!\le \!\frac{N}{4},\quad\\
	&\!\Big(\frac{M_1}{M_1+M_2K_2},0\Big), M_1+M_2K_2\!<\!N,\\
	&\!\Big(\frac{M_1}{N},\frac{1}{4}\Big),~ M_1\!+\!M_2K_2\!\ge \!N,\frac{N}{4}\!<\!M_1\!\le\! N. \qquad 
	\end{aligned}
	\right.
\end{equation}
\end{subequations}
\end{itemize}
\begin{figure}
	\centering
	\includegraphics[width=0.45\textwidth,scale=0.2]{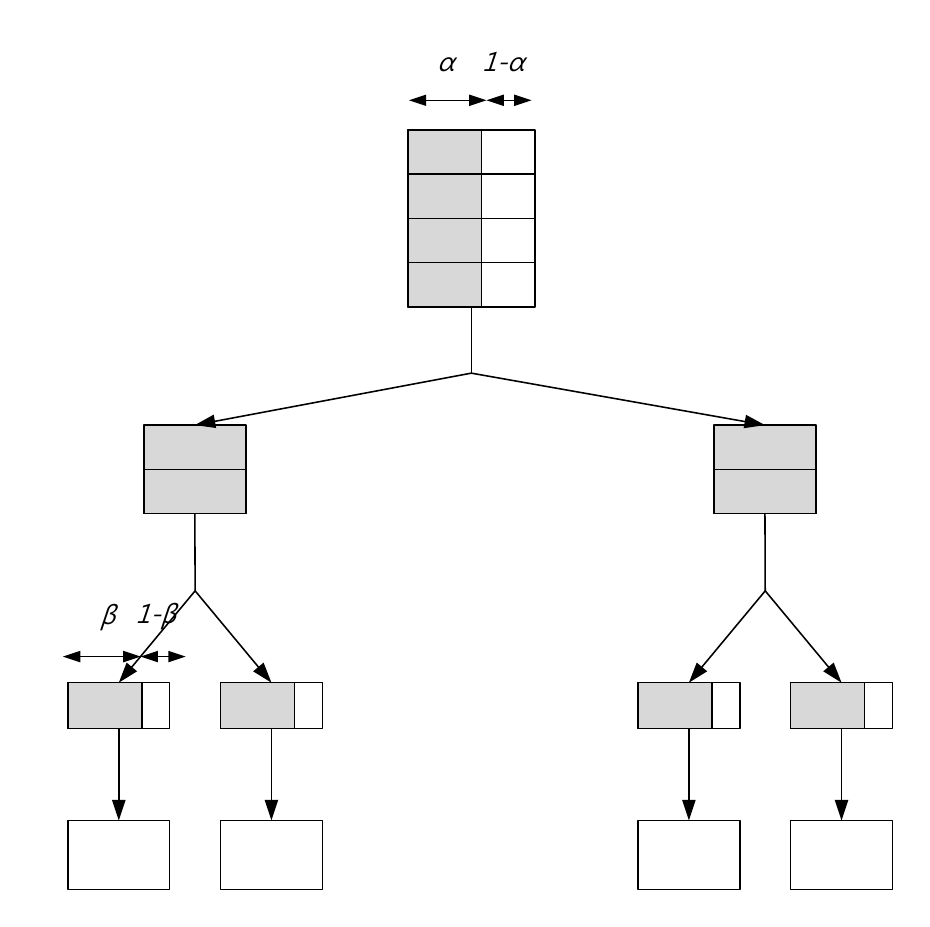}
	\caption{A two-layer caching system with two subsystems. For given $\alpha$ and $\beta$, the system is divided into two independent subsystems. Proposed caching scheme is applied to the first subsystem, and pipeline-forward  scheme is applied to the second subsystem. In this figure, $N=4$ and $K_1=K_2=2$.}\label{fig_hierarchical}
\end{figure}

Scheme  HCC-III reduces to scheme HCC-I when $\alpha=\beta=1$, and to scheme HCC-II when $\alpha=\beta=0$. Since  the transmission in two layers proceeds in a sequential progress, the transmission delay of scheme  HCC-III is 
\begin{IEEEeqnarray}{rCl}\label{eq_HccCRate}
T_{\mathsf{Hcc,D}}^{\text{original}}=R^{\text{\tiny III}}_{\mathsf{Hcc},1}+R^{\text{\tiny III}}_{\mathsf{Hcc},2}.
\end{IEEEeqnarray}
Here we use subscript $\mathsf{D}$ to represent the scheme that is based on decentralized coded caching scheme \cite{Decentralized}. 
In \cite{Karamchandani'16} it shows that  $R^{\text{\tiny III}}_{\mathsf{Hcc},1}$ and $R^{\text{\tiny III}}_{\mathsf{Hcc},2}$ achieve the optimal  rates within a constant multiplicative and additive gap. 

The idea of scheme HCC-I  can be easily extended  the setup under the centralized caching placement \cite{Centralized}. Define \begin{IEEEeqnarray}{rCl}
r_\mathsf{c}(\frac{M}{N},K) \triangleq \left[{K\left(1-\frac{M}{N}\right)}\frac{1}{1+{K M}/{N}}\right]^{+}.
\end{IEEEeqnarray} By combining the centralized HCC-I and HCC-II, we achieve the  transmission delay having same  expression as \eqref{eq_HccCRate},  but with  $r_\mathsf{d}(\cdot)$  in $(R^{\text{\tiny III}}_{\mathsf{Hcc},1}, R^{\text{\tiny III}}_{\mathsf{Hcc},2})$ replaced by $r_\mathsf{c}(\cdot)$, and $M_i\in\{0,{N}/{K_i},{2N}/{K_i},\ldots,N\}$, $i=1,2$. For general $0 \leq M \leq N$, the lower convex envelope of these points is achievable.


\section{Motivations and Examples}\label{Sec_Motivation}
We first propose a simple scheme, called \emph{pipeline forward}, to replace of the original HCC-II. New upper bounds on transmission delay are derived based on the  HCC-I and pipeline-forward scheme. Finally, we use  simple  examples to  illustrate that transmission delay can be greatly reduced by allowing concurrent transmission between the server and relays, which motivates our coded caching schemes. 
\subsection{Pipeline-Forward Scheme}

 In the \emph{pipeline-forward} scheme, each relay connects the two layers as one pipeline (the information flow moves smoothly through the pipeline), then the two-layer network  is equivalent to the single-layer  network in which a server connects with $K_1K_2$ users through a shared noiseless link.  In the following example, we will show that this simple pipeline-forward can even outperform scheme HCC-III. 
\begin{Example}[Pipeline forward] \label{eg_pipe}
Consider the two-layer network with a single relay ($K_1=1$) with $M_1=0,M_2<N$.   From \eqref{eq_HccARate}, we have $R_{\mathsf{Hcc},1}^{\text{\tiny I}}=K_1K_2$ and  $R_{\mathsf{Hcc},2}^{\text{\tiny I}}=r_\mathsf{d}(M_2/N,K_2)$. From \eqref{eq_HccBRate}, we have $R_{\mathsf{Hcc},1}^{\text{\tiny II}}=r_\mathsf{d}\left({M_2}/{N},K_1K_2\right)$ and $R_{\mathsf{Hcc},2}^{\text{\tiny II}}=r_\mathsf{d}(M_2/N,K_2)$. By \eqref{eq_HccCRate} we have $T_{\mathsf{Hcc,D}}^{\text{original}}> r_\mathsf{d}(M_2/N,K_1K_2)$. In fact, if we use the {pipelined-forward} scheme, then the server connects directly with $K_1K_2$ users, from \cite{Decentralized} the transmission delay is $r_\mathsf{d}(M_2/N,K_1K_2)$, which is always  smaller than  $T_{\mathsf{Hcc,D}}^{\mathsf{Original}}$.  The reason why the simple pipelined forward scheme outperforms  HCC scheme is that in scheme HCC, in particular HCC-I, each relay has to wait until it receives enough symbols (sent by the server) to  recover the requested files of all its attached users,  resulting in redundant transmission delay, while  the pipelined forward scheme saves the time by letting the relays and server consistently and concurrently deliver data. 
\end{Example}

\subsection{Modified HCC}\label{HCC centralized}
Combining the decentralized HCC-I with the  pipeline-forward scheme in a way same as HCC-III,  the upper bound of the transmission delay, denoted by  $T_\mathsf{Hcc,D}$, is
\begin{IEEEeqnarray}{rCl}\label{modifed_HCC}
	T_\mathsf{Hcc,D}\triangleq \alpha K_{2} r_\mathsf{d}\left(\frac{M_{1}}{\alpha N}, K_{1}\right)+\alpha r_\mathsf{d}\left(\frac{\beta M_{2}}{\alpha N}, K_{2}\right) \nonumber \\
	+(1-\alpha)r_\mathsf{d}\Big(\frac{(1-\beta)M_2}{(1-\alpha)N},K_1K_2\Big).
\end{IEEEeqnarray}
 Applying centralized caching strategy to  HCC-I and combine it  with the  pipeline-forward scheme,    the corresponding upper bound of the transmission delay, denoted by  $T_\mathsf{Hcc,C}$, has the same  expression as \eqref{modifed_HCC}, but with 
$r_\mathsf{d}(\cdot)$ replaced by $r_\mathsf{c}(\cdot)$, and  $M_i\in\{0,{N}/{K_i},{2N}/{K_i},\ldots,N\}$, $i=1,2$. For general $0 \leq M \leq N$, the lower convex envelope of these points is achievable.

Here the approximately optimal $(\alpha, \beta)$ for the modified HCC scheme   is given as 
\begin{IEEEeqnarray}{rCl}\label{choiceHCC}
	(\alpha_\mathsf{Hcc}^*, \beta_\mathsf{Hcc}^*)=\left\{\begin{array}{ll}{\left(\frac{M_{1}}{N,} \frac{M_{1}}{N}\right),} & {M_{1}+M_{2} K_{2} \geqslant N_{1}}, \\ {(\frac{M_{1}}{M_{1}+M_{2}K_{2}},0)}. & {M_{1}+M_{2} K_{2} < N_{1}},\end{array}\right.\quad
\end{IEEEeqnarray}
The selection of $\alpha$ and $\beta$  is very similar to the orignial HCC scheme, except that we merge the first and the third regime in \eqref{eqHCCab}. In  Appendix \ref{pr_ab}, we prove that  when  $K_{2} M_{2}/N$ is sufficiently large, $\alpha=\beta=M_1/N$ is optimal for $T_\mathsf{Hcc,C}$.

\subsection{Concurrent Transmission}
In scheme HCC-I,  each  relay  only starts broadcasting data   after it successfully decoding the   files requested by its attached users. In fact, each relay's cache memory may contain some pieces of files requested by its attached users, and thus we can save the transmission duration by letting the relay send these pieces of files to the users during the server's transmission.  Furthermore, scheme HCC-I  takes the relays as  ``tycoon users" and in the first layer the sever send files only with the help of relays' caches, completely ignoring the users' caches. This means that the server may send some redundant contents that have already been stored in users' cache memories.   Now consider the following examples.


\begin{Example}[Optimal case with $M_1=N$] Consider the two-layer network  with $N=4$ files ($A$, $B$, $C$ and $D$), $K_1=K_2=2$, $M_1=4$ and  $M_2=0$. This implies the  relays can access the full library and the users have zero caching capability. Obviously,  the optimal transmission delay of this case is  $T^*=2$. \label{eg_1}
\end{Example}

\begin{Example}[Concurrent transmission with $M_2=0$]\label{eg_2} 
Consider the two-layer network  with $N=4$ files ($A$, $B$, $C$ and $D$), $K_1=K_2=2$, $M_1=2$ and  $M_2=0$. For ease of explanation, we consider the centralized caching placement \cite{Centralized}. More specifically, each file is divided into two parts of equal size: $A=(A_1,A_2)$, $B=(B_1,B_2)$, $C=(C_1,C_2)$ and $D=(D_1,D_2)$. Relay $i=1,2$ caches $(A_i,B_i,C_i,D_i)$. User 1, 2, 3 and 4, request file $A$, $B$, $C$ and $D$, respectively.

Table \ref{table_Hcc} describes the delivery phase of scheme HCC-I. In  the first subphase of the deliver phase, the server sends  $A_2\oplus C_1$ and $B_2\oplus D_1$ in sequence.
In the second subphase, relay 1  sends files $A$ and $B$ in sequence, and at the same time, relay 2  sends files $C$ and $D$ in sequence. Thus,  scheme HCC-I achieves $R_1=1$, $R_2 =2$ and the transmission delay $3.$

\begin{table}[tbp]
\centering
\caption{Delivery Strategies of  HCC-A and a new scheme}
\begin{tabular}{c|cccc}
\toprule
&Server&Relay\ 1&Relay\ 2&Rate\\
\hline
\multirow{4}{*}{HCC-I}&$A_2\oplus C_1$&&&$1/2$\\
&$B_2\oplus D_1$&&&$1/2$\\
&&$A$&$C$&$1$\\
&&$B$&$D$&$1$\\
\hline
\multirow{4}{*}{New scheme}&$A_2\oplus C_1$&$A_1$&$C_1$&$1/2$\\
&$B_2\oplus D_1$&$B_1$&$D_1$&$1/2$\\
&&$A_2$&$C_2$&$1/2$\\
&&$B_2$&$D_2$&$1/2$\\
\bottomrule
\end{tabular}\label{table_Hcc}
\end{table}

Now we consider a new simple scheme as shown in Table \ref{table_Hcc}. In this scheme we let the relays send contents while receiving data from server, which achieves   the transmission delay $T=2<3$. This simple scheme improves scheme HCC-I since we 
exploit the transmission slots in a more efficient way such that  the relays and server concurrently deliver data. 


Comparing Example \ref{eg_1} with Example \ref{eg_2},  it's interesting to find thats our new scheme achieves the optimal coding  delay $T^*=2$, while only requiring $M_1=N/2=2$. That is,  even the relays lack enough cache memory to access the full library, it's feasible to achieve the transmission delay as if the relays had access to the full library.  

 \end{Example}

We conclude this section by listing the following insights:
\begin{itemize}
\item Concurrent transmission between the two layers can reduce the transmission delay.
\item For some cases, such as Example \ref{eg_2}, having partial size ($M_1<N$) of the library at the relays  may achieve the same transmission delay as the case $M_1=N$. In other words, enlarging the relay's cache memory  may not always reduce the transmission delay. 
\end{itemize}

\section{Main Results}\label{sec_results}
We now present a upper bound on the transmission delay of the network depicted in Fig.  \ref{fig_model}.

\begin{theorem}[Upper Bound of the Centralized Scheme] \label{delay of central}  For  all    $\alpha, \beta \in[0,1]$, and memory size $M_i\in\{0,{N}/{K_i},{2N}/{K_i},\ldots,N\}$, $i=1,2$, the optimal transmission delay $T^*$ is upper bounded by $T _\mathsf{Pro,C}$:  
\begin{subequations}\label{eqDelayCentral}
\begin{IEEEeqnarray}{rCl}
 T _\mathsf{Pro,C}\triangleq \alpha R_\mathsf{s1,C}+(1-\alpha)R_\mathsf{s2,C},
\end{IEEEeqnarray}
where 
\begin{IEEEeqnarray}{rCl}
R_\textnormal{s1,C} &\triangleq&R_\mathsf{p1,C} +R_\mathsf{p2,C}, \nonumber\\
R_\textnormal{s2,C} &\triangleq&r_\mathsf{c}\left(\frac{(1-\beta) M_{2}}{(1-\alpha) N},K_{1} K_{2}\right),\nonumber\\
	R_\mathsf{p1,C} &\triangleq&r_\mathsf{c}\left(\frac{M_{1}}{\alpha N},K_{1}\right) r_\mathsf{c}\left(\frac{\beta M_{2}}{\alpha N},K_{2}\right),\nonumber\\
	R_\mathsf{p2,C} &\triangleq& \min \left\{\frac{M_{1}}{\alpha N}, 1\right\} r_\mathsf{c}\left(\frac{\beta M_{2}}{\alpha N},K_{2}\right).
\end{IEEEeqnarray}
For general $0 \leq M \leq N$, the lower convex envelope of these points is achievable.

\end{subequations}
\end{theorem}
\begin{proof}  In the caching placement phase, we let the relays and users prefetch data to fill in their caches independently as in \cite{Centralized}. Notice that as the caching sizes of relays and users could be different, this leads to different sizes of subfiles stored at the relays and users, and in turn poses an obstacle on multicast transmission in the delivery phase.  Moreover, even if the relays' and users' caching sizes are the same, since each relay doesn't have all files, it cannot generate XOR symbols as the server in \cite{Centralized}. (The relay can wait and decode the server's signals until it recovers all required files,  then it can generate XOR symbol, but this turns to be scheme HCC-I and fails to make use of the parallel transmission between the server and relays.) To solve this problem, we partition the files into smaller subfiles of equal size, and design an elegant delivery strategy that still achieves the global multicast gain. Later we combine this proposed scheme with pipeline-forward scheme, and obtain the coding delay in  \eqref{eqDelayCentral}. See  detailed proof in Section \ref{sec_proof_centralized}.
\end{proof}

In Appendix \ref{pr_ab}, we show that the following  choice of $(\alpha,\beta)$ is approximately optimal:
\begin{IEEEeqnarray}{rCl}\label{choiceOur}
	\alpha^* = \beta^*  =\min\{\gamma,1\}
\end{IEEEeqnarray}
where  
\begin{IEEEeqnarray}{rCl}
					\gamma &\triangleq&\frac{2 K_{1} K_{2} M_{1} M_{2}+K_{1} M+M_{1} N}{2 K_{1} K_{2} M_{2} N-2 K_{2} M_{2} N}.
					\end{IEEEeqnarray}

\begin{theorem}\label{theorem_rateDec}
For all $M_1,M_2\in[0,N]$ and $\alpha,\beta\in[0,1]$, the optimal transmission delay $T^*$ is upper bounded by  $T_\mathsf{Pro,{D}}$:
 \begin{subequations}\label{eq_proposeRate}
\begin{IEEEeqnarray}{rCl}
 T_\mathsf{Pro,{D}}\triangleq \alpha R_\mathsf{s1,D}+(1-\alpha)R_\mathsf{s2,D}
\end{IEEEeqnarray}
where  
\begin{IEEEeqnarray}{rCl}
R_\mathsf{s1,D} &\triangleq &R_\mathsf{p1,D}+R_\mathsf{p2,D}+\mathsf{max}\{R_\mathsf{p3,D}-R_\mathsf{e,D}, 0\},\nonumber\\
R_\mathsf{s2,D}&\triangleq & r_{\mathsf{d}}\Big(\frac{(1-\beta)M_2}{(1-\alpha)N},K_1K_2\Big),
\end{IEEEeqnarray}
with \begin{IEEEeqnarray}{rCl}
R_\mathsf{e,D}&\triangleq &\Big(1-\frac{M_1}{\alpha N}\Big)^{\!K_1} \!\!\Big(1-\frac{\beta M_2}{\alpha N}\Big)^{\!K_2} \! r_{\mathsf{d}}\Big(\frac{\beta M_2}{\alpha N},(K_1\!-\!1)K_2\Big),\nonumber\\
R_\mathsf{p1,D}&\triangleq&\nonumber\left[r_{\mathsf{d}}\Big(\frac{M_1}{\alpha N},K_1\right)\!-\!K_1\left(1\!-\!\frac{M_1}{ \alpha N}\Big)^{\!K_1}\right] r_{\mathsf{d}}\Big(\frac{\beta M_2}{\alpha N},K_2\Big),\\
R_\mathsf{p2,D}&\triangleq&\Big(1-\frac{M_1}{\alpha N}\Big)^{K_1} r_{\mathsf{d}}\Big(\frac{\beta M_2}{\alpha N},K_1K_2\Big),\nonumber\\
R_\mathsf{p3,D}&\triangleq&\min \left\{\frac{M_{1}}{\alpha N}, 1\right\} r_{\mathsf{d}}\Big(\frac{\beta M_2}{\alpha N},K_2\Big).
\end{IEEEeqnarray}

\end{subequations}
\end{theorem}
\begin{proof}
The decentralized placement procedure is applied to cache memories in relays and users independently. The delivery strategy is complicated since the subfiles stored in  different relays have  different impact on the transmission, and in order to fully exploit the opportunities of   multicasting and parallel transmission, we need to carefully design how to encode the subfiles and how to send them parallelly at the server and  relays.  For each user $u^i_j\in\mathcal{U}$, we divide the required  subfiles into three parts: Subfiles \uppercase\expandafter{\romannumeral 1} are  those cached by other relays except relay $i$, and will be sent by a strategy similar to the decode-forward scheme\cite{Cover'79};   Subfiles \uppercase\expandafter{\romannumeral 2} are those not cached by any relay, and will be send by  the  pipelined forward  strategy introduced in Example \ref{eg_pipe}; Subfiles \uppercase\expandafter{\romannumeral 3} are those cached by relay $i$, and will be sent by the single-layer decentralized caching scheme.  When the server sends  Subfiles \uppercase\expandafter{\romannumeral 2}, there are some parts which are redundant for relay $i$. We let  relay $i$ send parts of Subfiles III  if the server's signal is not useful. Another caching scheme is to  simply use the  pipeline forward  strategy to send all the requested files.  Combining these two caching schemes, we obtain the transmission delay in  \eqref{eq_proposeRate}. See  detailed proof in Section \ref{sec_proof_decentralized}.
\end{proof}

\begin{theorem}[Wasteful Cost at $M_1$] \label{UselessM1} When using the centralized caching placement as \cite{Centralized}, if $M_1$ equals to the threshold $(K_1-1)N/{K_1}$,  our scheme  achieves  the optimal transmission delay  same as $M_1=N$, i.e., 
\[T_\mathsf{Pro,C}=r_\mathsf{c}\left(\frac{M_{2}}{N},K_{2}\right).\]

For the two-relay case ($K_1=2$), when using the decentralized caching placement as \cite{Decentralized}, if 
  $M_1$ equals to the threshold $0.382N$,  our scheme  achieves  the optimal transmission delay  same as $M_1=N$, i.e., 
\[T_\mathsf{Pro,D}=r_\mathsf{d}\left(\frac{M_{2}}{N},K_{2}\right).\]
\end{theorem}
\begin{proof}
See proof in Appendix \ref{proof_TheoremUselssCen}.
\end{proof}

 Theorem \ref{UselessM1} indicates that increasing caching size of relay  will not always reduce the transmission delay, even when relay's caching size is strictly smaller than the size of library. To our best knowledge, this is the first result  showing that increasing caching size (for non-trivial case)  is not helpful in reducing transmission latency.  

The following theorem presents  the lower bound on the transmission delay.
\begin{theorem}\label{eq_lowerbound}
For all $M_1,M_2\in[0,N]$, $s_1\in[K_1]$ and $s,s_2\in[K_2]$,  the optimal transmission delay $T^*$ is lower bounded by
\begin{IEEEeqnarray}{rCl}\label{eq1_lowerbound}
T^*\geq  \max\left\{ s_1s_2 -\frac{s_1M_1+s_1s_2M_2}{\lfloor N/(s_1s_2) \rfloor}, s-\frac{sM_2}{\lfloor N/s \rfloor}\right\}.\quad 
\end{IEEEeqnarray}
\end{theorem}
\begin{proof}
The first term on right hand of \eqref{eq1_lowerbound} follows from the similar cut-set bound given in   \cite[Appendix A]{Karamchandani'16}.  The second term is obtained by the cut-set bound  assuming $M_1=N$. In this case, the relays  access the  full library and the two-layer network is equivalent to the single-layer network where a server connects with $K_2$ users each caching $M_2$ files, and thus from \cite[Theorem 8]{Centralized}, we   obtain the second term on right hand of \eqref{eq1_lowerbound}. 
\end{proof}
Comparing the upper bounds above with the lower bound  \eqref{eq1_lowerbound}, we have
\begin{theorem}\label{theorem_compare}
For all $M_1,M_2\in[0,N]$, 
\begin{subequations}
\begin{IEEEeqnarray}{lll}\label{eq_GapHCC}
&T_{\mathsf{Hcc,D}}^{\mathsf{Original}}-T_\mathsf{Pro,D} \geq\nonumber\qquad&\\
&\quad
\left\{\begin{aligned}
	&\Big(1-\frac{M_1}{N}\Big) r_\mathsf{d}\Big(\frac{M_2}{N},K_2\Big), \hspace{13ex}\textnormal{regime~\uppercase\expandafter{\romannumeral 1}},\\
	&\frac{M_2K_2}{M_1+M_2K_2} r_\mathsf{d}\Big(\frac{M_1+M_2K_2}{NK_2},K_2\Big),  \hspace{1.5ex}\textnormal{regime~\uppercase\expandafter{\romannumeral 2}},\\
	&\Big(1-\frac{M_1}{N}\Big) r_\mathsf{d}\Big(\frac{3M_2}{4(N-M_1)},K_2\Big),  \hspace{5ex}\textnormal{regime~\uppercase\expandafter{\romannumeral 3}},  \qquad
	\end{aligned}
	\right.
\end{IEEEeqnarray}
and
\begin{IEEEeqnarray}{rCl}
 T_\mathsf{Pro,C}\leq T_\mathsf{Hcc,C} \leq c_{1}\cdot T^*, \label{eqCompareC}\\  
   T_\mathsf{Pro,D}\leq T_\mathsf{Hcc,D} \leq c_{2}\cdot T^* \label{eqCompareD}
\end{IEEEeqnarray}
where  regime 1, 2 and 3 represent $M_1+M_2K_2\ge N , 0\le M_1\le \frac{N}{4}$, $M_1+M_2K_2<N$ and $M_1+M_2K_2\ge N,\frac{N}{4}\!<\!M_1\le N$, respectively, and $c_{2}, c_{2}$ are finite positive constants independent of all the problem parameters.
\end{subequations}
\end{theorem}
\begin{proof}	
See proof in Appendix \ref{prTheo1}.
\end{proof}{
In \cite{Karamchandani'16}, the authors showed that in regime  $M_1+M_2K_2\ge N$, by using their schemes,  the rate transmitted in the first layer has a constant multiplicative gap of 35 within $T^*$. Since in their schemes, especially Scheme HCC-I, the first and second layer transmit in a sequential manner,    the gap between the their upper bounds and $T^*$ could be even larger than 35. In Appendix \ref{prTheo1}, we show that in regime  $M_1+M_2K_2\ge N$, our proposed upper bound is within a  multiplicative gap of 24. 

 As we will see in the schemes described in Section \ref{sec_proof_centralized} and \ref{sec_proof_centralized}, the main improvement of our schemes mainly comes from two facts: 1) we fully exploit the time resource by allowing the server and relays to parallelly transmit signals; 2) after  receiving the coded  package from the server, each relay doesn't need to decode every subfile contained in the XOR symbols. It can decode part of XOR symbols and send it to the attached users, by which the users can decode their desired subfiles  using their cached contents.


{Fig. \ref{ratePlot1} and \ref{ratePlot2} plot the lower bound  \eqref{eq_lowerbound}, together with  upper bounds  of various schemes including the modified decentralized and centralized HCC schemes (combined with the pipeline-forward scheme), and our proposed centralized and decentralized  schemes.    
It can be seen that the proposed schemes lead to lower transmission delay compared to the HCC schemes. In Fig. \ref{ratePlot1}, as shown in the dot-dash line (decentralized caching strategy), when $M_1=500=0.25N$, increasing $M_1$ can not reduce the transmission delay; as shown in the  line with stars, when $M_1=1000=0.5N$, increasing $M_1$ can not reduce the transmission delay. This  coincides with the results in Theorem \ref{UselessM1} showing that increasing  caching size of relay  will not always reduce the transmission delay. 
}

\begin{figure}
           \centering
	\includegraphics[width=0.5\textwidth]{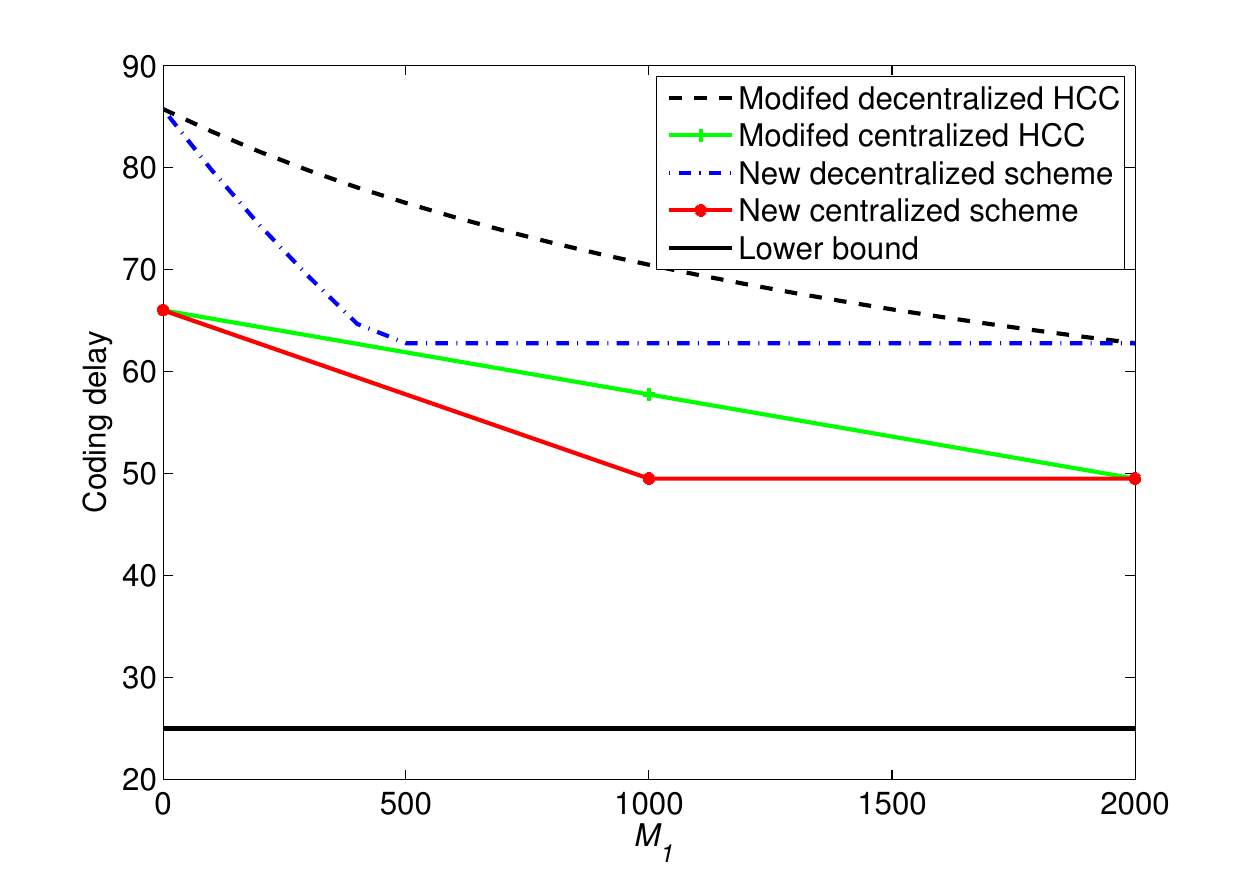}
	\caption{Bounds on $T^*$ for $N=2000$, $M_2=20$, $K_1=2$, $K_2=100$.}
	\label{ratePlot1}
\end{figure}

\begin{figure}
           \centering
	\includegraphics[width=0.5\textwidth]{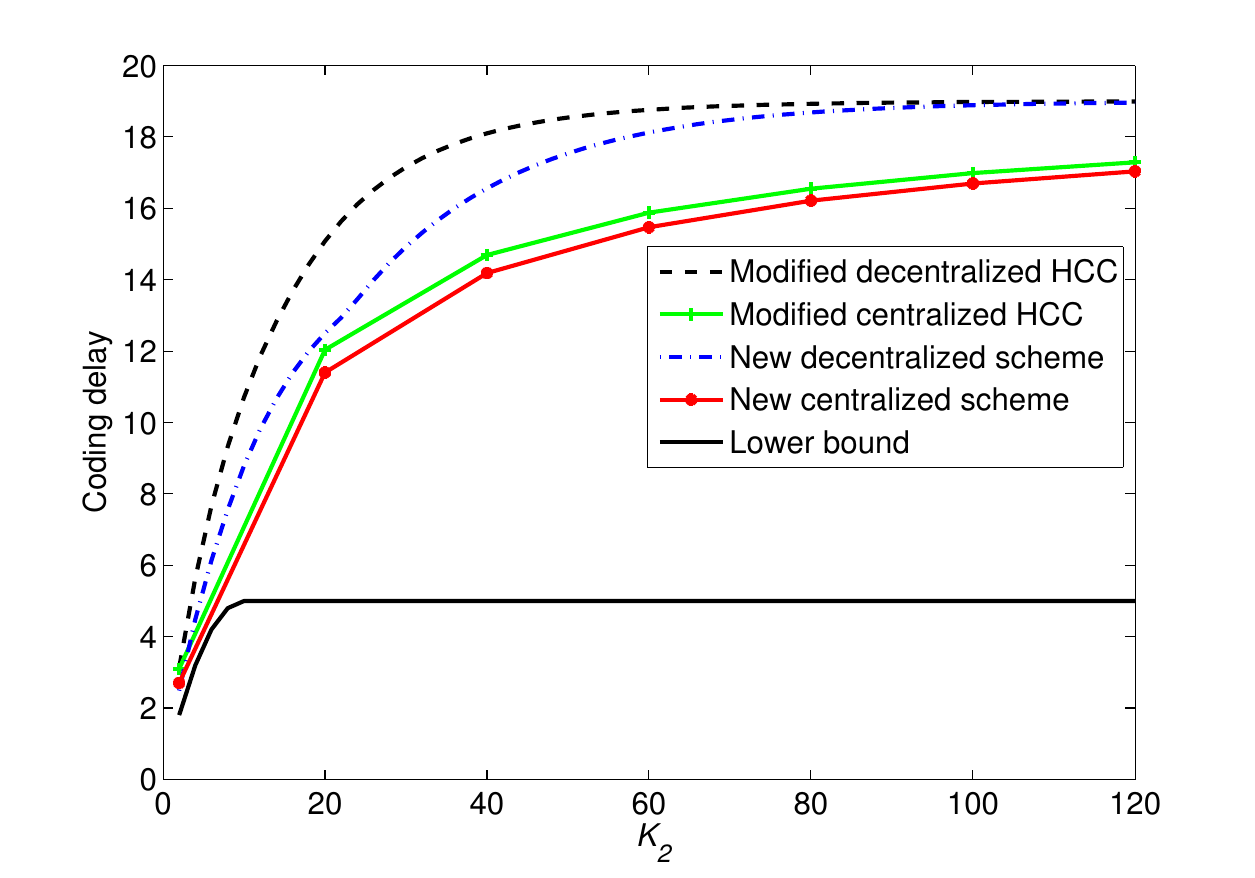}
	\caption{Bounds on $T^*$ for $N=2000$, $M_2=100$, $K_1=2$, $M_1=400$.}
	\label{ratePlot2}
\end{figure}

\section{Proof of Theorem \ref{delay of central}}\label{sec_proof_centralized}
In this section, we first present a centralized coded caching scheme for the two-layer network, and then combine it with pipeline-forward scheme, which leads to a upper bound on the transmission delay shown in Theorem \ref{delay of central}.

\subsection{Concurrent Centralized Caching Scheme}\label{centralized_nomix}

For any $K$, $N$ and cache size $M_1 \in \{0, N/K, 2N/K,...,N\}$, $M_2 \in \{0, N/K, 2N/K,...,N\}$, define
\[t_1 \triangleq\frac{K_1M_1}{N}, \quad t_2\triangleq\frac{K_2M_2}{N}.\]  Note that $t_1$ is an integer between 0 and $K_1$, and $t_2$  between 0 and $K_2$. 

In the placement phase, we split each file into $\binom{K_1}{t_1}\binom{K_2}{t_2}$ subfiles of equal size. Index the subfiles of $W_n$ by the superscript $\mathcal{Q} \in [K_1]$ and subscript $\mathcal{T} \in [K_2]$:
\begin{IEEEeqnarray}{rCl}
	W_n = \left(W_{n,\mathcal{T}}^{\mathcal{Q}}: \mathcal{Q} \in [K_1],\mathcal{T} \in [K_2], |\mathcal{Q}|=t_1,|\mathcal{T}|=t_2 \right).
	\nonumber
\end{IEEEeqnarray}
Relay $i$ caches all the subfiles when $i \in \mathcal{Q}$ and user $u^i_j$ stores all the subfiles when $j \in \mathcal{T}$ for all $n=1,...,N$. The above placement strategy requires each relay to store 
\[
N \cdot \frac{F}{\binom{K_1}{t_1}}\cdot \binom{K_1-1}{t_1-1}=F\cdot \frac{Nt_1}{K_1}=M_1F
\]
bits of files and each user to cache
\[
N \cdot \frac{F}{\binom{K_2}{t_2}}\cdot \binom{K_2-1}{t_2-1}=F\cdot \frac{Nt_2}{K_2}=M_2F
\]
bits of data, satisfying the cache size constraints for both relays and users.

Note that according to the prefetching strategy, users indexed by $u_j^i$ with different superscript $i$ but the same subscript $j$ prefetch same contents in their local cache memories. 

In the delivery phase, each user $u_j^i$ reveals its request $W_{d_j^i}$ to the server. The requests vector $\textbf{d}$ are informed by the server and all the relays during this phase. Our objective is to get the upper bound on transmission delay in the worst request case, so we assume that each user makes unique request in the following discussion. Since different parts of the file $W_{d_j^i}$ have been stored in the relays' caches and the users' caches, the subfiles needed by user $u_j^i$ can be divided into two parts: 

\begin{itemize}
\item Subfile \uppercase\expandafter{\romannumeral 1}: 
subfiles cached by other relays except relay $i$, i.e., 
$W_{d_j^{~\!\!i},\mathcal{T}}^{\mathcal{Q}}$, for all $ \mathcal{Q}\subseteq[K_1]:|\mathcal{Q}|=t_1$, $i\notin \mathcal{Q}$  and  $\mathcal{T} \subset [K_2]: |\mathcal{T}|=t_2$, $j\notin \mathcal{T}$.
\item subfile \uppercase\expandafter{\romannumeral 2}:
subfiles cached by relay $i$, i.e. $W_{d_j^{~\!\!i},\mathcal{T}}^{\mathcal{Q}}$, for $ \mathcal{Q}\subseteq[K_1]:|\mathcal{Q}|=t_1$, $i\in\mathcal{Q}$ and  $\mathcal{T} \subset [K_2]: |\mathcal{T}|=t_2$, $j\notin \mathcal{T}$.

\end{itemize}
Subfile \uppercase\expandafter{\romannumeral 1} have to be sent by the server to relay $i$. After receiving the server transmission, relay $i$ decodes the interested contents and forward the data to its attached user $u_j^i$. Subfile \uppercase\expandafter{\romannumeral 2}  will be sent by relay $i$ directly to its attached users. For the two different kinds of subfiles, in the followed discussion, we illustrate two kinds of data transmissions, which are labelled as Transmission \uppercase\expandafter{\romannumeral 1} and  \uppercase\expandafter{\romannumeral 2}, respectively. 
\begin{itemize}
	
	\item Transmission \uppercase\expandafter{\romannumeral 1}: Since there is no communication or cooperation among  the relays,  Subfiles \uppercase\expandafter{\romannumeral 1} can only be offered by the server. Note that for relay $i \in [K_2]$, all its attached users $u_j^i, j=1,...,K_2$ have worked as a group and cached the whole library of the server to their local cache memories distributedly. Besides, for any relay $k\in [K_1], i\neq k$, user $u_j^k$ has cached the same contents to local cache memory as user $u_j^i$. Therefore, the server can create signals that are useful for multiple groups of users connected with different relays simultaneously. 
	
	For each $\mathcal{R} \subset [K_1], |\mathcal{R}|=t_1+1$ and $\mathcal{S} \subset [K_2],|\mathcal{S}|=t_2+1$, the server sends the symbol
	\begin{IEEEeqnarray}{rCl}\label{CentralizedSeverSignal}
		\oplus_{i\in \mathcal{R}}\left( \oplus_{j\in\mathcal{S}} W_{d^i_j,\mathcal{S}\backslash \{j\}}^{\mathcal{R}\backslash\{i\}} \right)
	\end{IEEEeqnarray}
	to all the relays over the shared link, each of size $F/\binom{K_1}{t_1}\binom{K_2}{t_2}$ bits. 
	
	The relays use a block-Markov partial decode-forward strategy \cite{Cover'79}. More specifically,  assuming the whole transmission takes in $B+1$ blocks, at the beginning of block $b$, each relay starts collecting XOR symbols of block $b\in{1,...,B+1}$ sent from the server, and at the same time it starts sending the decoded XOR symbols of block $b-1$.  \footnote{The transmission in fact needs  some extra delay for sending the initial symbol, as in block $b=0$, the relays haven't receive any data from server. This extra delay can be ignored if the number of symbols is sufficiently large, or by letting the relays send subfiles stored in their cache memory to the users without waiting the server's transmission. } Thus, once relay $i$ receives the server transmission \eqref{CentralizedSeverSignal} , it decodes the interested XOR subfiles 
	\begin{IEEEeqnarray}{rCl}\label{CentralizedRelaySignal2}
		\oplus_{j\in\mathcal{S}} W_{d^i_j,\mathcal{S}\backslash \{j\}}^{\mathcal{R}\backslash\{i\}}
	\end{IEEEeqnarray}
	from \eqref{CentralizedSeverSignal} and forwards it to the attached users immediately. After observing the message \eqref{CentralizedRelaySignal2}, each user $u^i_j:j\in\mathcal{S}$ decodes $	W_{d^i_j,\mathcal{S}\backslash \{j\}}^{\mathcal{R}\backslash\{i\}}$ based on the cached content. 
	
	Consider a fixed relay $i\in \mathcal{R}:|\mathcal{R}|=t_1+1$  and user $u_j^i\in \mathcal{S}:|\mathcal{S}|=t_2+1$, there are $\binom{K_1-1}{t_1}$ partitions to form $\mathcal{R}$ and $\binom{K_2-1}{t_2}$ partitions to form $\mathcal{S}$, respectively. The rate $\bar{R}_\mathsf{p1,C}$ of Transmission \uppercase\expandafter{\romannumeral 1} is 
	\begin{equation}
	\begin{aligned}
	 \bar{R}_\mathsf{p1,C}&=K_1K_2 \cdot \frac{ \binom{K_1-1}{t_1}\binom{K_2-1}{t_2}}{(t_1+1)(t_2+1)} \cdot \frac{F}{\binom{K_1}{t_1}\binom{K_2}{t_2}}\cdot\frac{1}{F}\\
	&= \frac{(K_1-t_1)(K_2-t_2)}{(t_1+1)(t_2+1)}.
	\end{aligned}
	\end{equation}

	\item Transmission \uppercase\expandafter{\romannumeral 2}:  All the  relays create multicasting opportunities and send signals concurrently to their attached users. For relay $i \in \mathcal{R}: |\mathcal{R}|=t_1$, consider a subset $\mathcal{S} \subset [K_2]$ of  $|\mathcal{S}|=t_2+1$ users attached to it.  Observe a user $u_j^i:j\in \mathcal{S}$. Subfile $W_{d^i_j, \mathcal{S}\backslash\{j\}}^{\mathcal{R}}$ that is required by user $u_j^i$ is
	present in the cache of any user $u^i_k: k\in \mathcal{S}_i\backslash\{j\}$. Thus, for each subset $\mathcal{S} \subset [K_2]$ of size $|\mathcal{S}|=t_2+1$, relay $i$ multicasts 
	\begin{equation}
	\oplus_{j\in\mathcal{S}}W_{d^i_j, \mathcal{S}\backslash\{j\}}^{\mathcal{R}}\nonumber
	\end{equation}
	for each $j \in [K_2]$. Each transmission results in $F/ \binom{K_1}{t_1}\binom{K_2}{t_2}$ bits of data being sent over the shared link. Since there are $\binom{K_1-1}{t_1-1}$ different selections for $\mathcal{R}$ and the number of subset $\mathcal{S}$ is $\binom{K_2-1}{t_2}$,  the transmission rate for sending  Subfile \uppercase\expandafter{\romannumeral 2}, denoted by $\bar{R}_\mathsf{p2,C}$, can be summarized as follows: 
	\begin{equation}
	\begin{aligned}
	\bar{R}_\mathsf{p2,C}&=K_2 \cdot \frac{ \binom{K_1-1}{t_1-1}\binom{K_2-1}{t_2}}{t_2+1} \cdot \frac{F}{\binom{K_1}{t_1}\binom{K_2}{t_2}}\cdot\frac{1}{F}\\
	&= \frac{t_1}{K_1}\frac{K_2-t_2}{t_2+1}.
	\end{aligned}
	\end{equation}
	 
\end{itemize}
Since the server and relays are allowed to transmit  signals concurrently, there are three scenarios when the delivery phase begins:
\begin{itemize}
	\item $\bar{R}_\mathsf{p1,C}<\bar{R}_\mathsf{p2,C}$: When the size of Subfile \uppercase\expandafter{\romannumeral 2} is larger than Subfile \uppercase\expandafter{\romannumeral 1}, Transmission \uppercase\expandafter{\romannumeral 1} has been finished before the relays send all Subfile \uppercase\expandafter{\romannumeral 2}  to users. Then the relays decode the required subfiles from the server transmission and forward the messages to corresponding users. The transmission delay of the proposed caching scheme is 
	\[
	\bar{R}_\mathsf{s1,C}=\bar{R}_\mathsf{p1,C}+\bar{R}_\mathsf{p2,C}.
	\]   
	\item $\bar{R}_\mathsf{p1,C}=\bar{R}_\mathsf{p2,C}$: In this scenario, Transmission \uppercase\expandafter{\romannumeral 1} and \uppercase\expandafter{\romannumeral 2} finish synchronically. Then the relays forward the decoded messages to their attached users, resulting in the transmission delay to be
	\[
	\bar{R}_\mathsf{s1,C}=\bar{R}_\mathsf{p1,C}+\bar{R}_\mathsf{p2,C}.
	\] 
	\item $\bar{R}_\mathsf{p1,C}>\bar{R}_\mathsf{p2,C}$: When the size of Subfile \uppercase\expandafter{\romannumeral 1} is larger than Subfile \uppercase\expandafter{\romannumeral 2}, Transmission \uppercase\expandafter{\romannumeral 2} ends earlier. When Transmission \uppercase\expandafter{\romannumeral 2} finishes, $F\bar{R}_\mathsf{p2,C}$ bits of Subfile \uppercase\expandafter{\romannumeral 1} have been sent to the relays, while $F(\bar{R}_\mathsf{p1,C}-\bar{R}_\mathsf{p2,C})$ bits of them are still needed to be sent. All the relays then decode  and forward the useful parts of messages to their users, while they are receiving the server transmission simultaneously. After observing the  extra $F(\bar{R}_\mathsf{p1,C}-\bar{R}_\mathsf{p2,C})$ bits of Subfile \uppercase\expandafter{\romannumeral 1}, the relays operate in the same way, making the transmission delay to be
	\[
	\begin{aligned}
	\bar{R}_\mathsf{s1,C}&=\bar{R}_\mathsf{p2,C}+\bar{R}_\mathsf{p2,C}+(\bar{R}_\mathsf{p1,C}-\bar{R}_\mathsf{p2,C}).\\
	&=\bar{R}_\mathsf{p1,C}+\bar{R}_\mathsf{p2,C}.
	\nonumber
	\end{aligned}
	\]
\end{itemize}
According to the analysis above and \eqref{CentralizedSeverSignal}, \eqref{CentralizedRelaySignal2}, the transmission delay of the proposed caching scheme is 
\begin{equation}\label{eq_cen_nomix} 
\bar{R}_\mathsf{s1,C}=\bar{R}_\mathsf{p1,C}+\bar{R}_\mathsf{p2,C}.
\end{equation}

The coded caching scheme  is summarized in Algorithm 1. In order to explain the key steps in scheme described above, we will then present a simple example.

\begin{algorithm}\label{algoCentra}
	\caption{Centralized Caching for Two-Layer Network}
	\begin{algorithmic}[1]
		\STATE{\textbf{Procedure 1} PLACEMENT $(W_1,...,W_N)$}
		\STATE{$t_1 \leftarrow K_1M_1/N$, $t_2 \leftarrow K_2M_2/N$}
		\STATE{$\mathfrak{Q}\leftarrow \{ \mathcal{Q}\subset [K_1]: |\mathcal{Q}|=t_1\}$}
		\STATE{$\mathfrak{T}\hspace{0.7mm}\leftarrow \{ \mathcal{T}\subset [K_2]: |\mathcal{T}|=t_2\}$}
		\FOR{$n\in [N]$}
		\STATE{Split $W_n$ into $(W_{n,\mathcal{T}}^{\mathcal{Q}}:\mathcal{Q} \in \mathfrak{Q}, \mathcal{T} \in \mathfrak{T})$ of equal size}
		\ENDFOR
		\FOR{$i \in [K_1]$}
		\STATE{$Z_i\leftarrow (W_{n,\mathcal{T}}^{\mathcal{Q}}:n\in [N], \mathcal{Q} \in \mathfrak{Q}, i \in \mathcal{Q})$}
		\FOR{$j \in [K_2]$}
		\STATE{$Z^i_j\leftarrow (W_{n,\mathcal{T}}^{\mathcal{Q}}:n\in [N], \mathcal{T} \in \mathfrak{T}, j \in \mathcal{T})$}
		\ENDFOR
		\ENDFOR
		\STATE{\textbf{End Procedure}}
		\STATE{}
		\STATE{\textbf{Procedure 2} DELIVERY \uppercase\expandafter{\romannumeral 1} $(W_1,...,W_N, \textbf{d})$}
		\STATE{$t_1 \leftarrow K_1M_1/N$, $t_2 \leftarrow K_2M_2/N$}
		\FOR{$i \in [K_1]$}
		\FOR{$\mathcal{R}\subset [K_1]:|\mathcal{R}|=t_1, i \in \mathcal{R}$}
		\FOR{$\mathcal{S}\subset [K_2]:|\mathcal{S}|=t_2+1$}
		\STATE{$X_{i,\textbf{d}}\leftarrow 	\oplus_{j\in\mathcal{S}}W_{d^i_j, \mathcal{S}\backslash\{j\}}^{\mathcal{R}}$}
		\ENDFOR
		\ENDFOR 
		\ENDFOR
		\STATE{\textbf{End Procedure}}
		\STATE{}
		\STATE{\textbf{Procedure 3} DELIVERY \uppercase\expandafter{\romannumeral 2} $(W_1,...,W_N, \textbf{d})$}
		\STATE{$t_1 \leftarrow K_1M_1/N$, $t_2 \leftarrow K_2M_2/N$}
		\FOR{$i\in[K_1]$}
		\FOR{$\mathcal{R}\subset [K_1]:|\mathcal{R}|=t_1+1, i \in \mathcal{R}$}
		\FOR{$\mathcal{S}\subset [K_2]:|\mathcal{R}|=t_2$}
		\STATE{$X_{\textbf{d}}\leftarrow 		\oplus_{i\in \mathcal{R}}\left( \oplus_{j\in\mathcal{S}} W_{d^i_j,\mathcal{S}\backslash \{j\}}^{\mathcal{R}\backslash\{i\}} \right)$}
		\ENDFOR
		\ENDFOR 
		\ENDFOR
		\STATE{\textbf{End Procedure}}
		
	\end{algorithmic}
\end{algorithm}

\begin{Example}[Concurrent transmission with $M_2\neq 0$]\label{central_eg}
	Consider a network consisting of $K_1=3$ relays, each connected with $K_2=2$ users, and a library of $N=6$ files. Each relay is equipped with a cache of size $M_1=4$, while each user's cache size is $M_2=3$. 
	For ease of notation, assume that user $u^1_1,u^1_2, u^2_1,u^2_2,u^3_1, u^3_2$ want file $W_1,W_2,W_3,W_4,W_5,W_6,$ respectively. 
	
	In the placement phase, each file is split into $\binom{K_1}{t_1}\binom{K_2}{t_2}=6$ subfiles of equal size. The content placement of each relay is
	\begin{equation}
	\begin{aligned}
	&Z_1=(W_{n,\{1\}}^{\{12\}},W_{n,\{2\}}^{\{12\}},W_{n,\{1\}}^{\{13\}},W_{n,\{2\}}^{\{13\}})_{n=1}^6,\\
	&Z_2=(W_{n,\{1\}}^{\{12\}},W_{n,\{2\}}^{\{12\}},W_{n,\{1\}}^{\{23\}},W_{n,\{2\}}^{\{23\}})_{n=1}^6,\\
	&Z_3=(W_{n,\{1\}}^{\{13\}},W_{n,\{2\}}^{\{13\}},W_{n,\{1\}}^{\{23\}},W_{n,\{2\}}^{\{23\}})_{n=1}^6.	
	\end{aligned}
	\end{equation}
	And the content placement of each user is
	\begin{equation}
	\begin{aligned}
	&Z_1^1=Z_1^2=Z_1^3=(W_{n,\{1\}}^{\{12\}},W_{n,\{1\}}^{\{13\}},W_{n,\{1\}}^{\{23\}})_{n=1}^6,\\ &Z_2^1=Z_2^2=Z_2^3=(W_{n,\{2\}}^{\{12\}},W_{n,\{2\}}^{\{13\}},W_{n,\{2\}}^{\{23\}})_{n=1}^6.
	\end{aligned}
	\end{equation}
	
	During the delivery phase,  relay 1, 2 and 3 multicast XOR symbols 
	\[
	\begin{aligned}
	&W_{1,\{2\}}^{\{12\}}\oplus W_{2,\{1\}}^{\{12\}}, ~W_{1,\{2\}}^{\{13\}}\oplus W_{2,\{1\}}^{\{13\}};\\
	&W_{3,\{2\}}^{\{12\}}\oplus W_{4,\{1\}}^{\{23\}},~
	W_{3,\{2\}}^{\{12\}}\oplus W_{4,\{1\}}^{\{23\}};\\
	&W_{5,\{2\}}^{\{13\}}\oplus W_{6,\{1\}}^{\{23\}},~
	W_{5,\{2\}}^{\{13\}}\oplus W_{6,\{1\}}^{\{23\}}
	\end{aligned}
	\]
	respectively to their attached users, resulting in transmission rate   of $\bar{R}_\mathsf{p2,C}={1}/{3}$. At the  same time, the server sends 
	\[\small
	\begin{aligned}
	\big( W_{1,\{2\}}^{\{23\}}\oplus W_{2,\{1\}}^{\{23\}} \big) \oplus \big( W_{3,\{2\}}^{\{13\}}\oplus W_{4,\{1\}}^{\{13\}} \big) \oplus \big( W_{5,\{2\}}^{\{12\}}\oplus W_{6,\{1\}}^{\{12\}} \big)
	\end{aligned}
	\] 
	to all the relays. The rate of Transmission \uppercase\expandafter{\romannumeral 1} is $\bar{R}_\mathsf{p1,C}={1}/{6}$. After observing the server transmission, relay 1 decodes $\big( W_{1,\{2\}}^{\{23\}}\oplus W_{2,\{1\}}^{\{23\}} \big)$ and forward it to the users. Relay 2 and 3 operate in the same way and send $\big( W_{3,\{2\}}^{\{13\}}\oplus W_{4,\{1\}}^{\{13\}} \big)$ and $ \big( W_{5,\{2\}}^{\{12\}}\oplus W_{6,\{1\}}^{\{12\}} \big)$ to the corresponding users, respectively. The total transmission delay  is thus ${1}/{2}$ for this example. From this example,  we see that after receiving the  packages of subfiles from the server, relay doesn't need to decode it completely. Users with cache can process the last decoding work, reducing the transmission delay to a certain extent.
\end{Example}

\subsection{Hybrid Centralized Scheme}
Apply the similar method described in [3, Sec. \uppercase\expandafter{\romannumeral 5}-C] to combine the scheme described above with pipeline-forward scheme. Divide the system model into two subsystems with two parameters $\alpha,\beta\in[0,1]$. The first subsystem includes the entire cache memory in each relay and a $\beta$ fraction of each user's cache memory, and the second subsystem holds the remaining $1-\beta$ fraction of each user's cache memory. Additionally, each file is split into two parts of size $\alpha F$ and size $(1-\alpha)F$ bits, as showed in Fig. \ref{fig_hierarchical}. After that, the proposed caching scheme in Section \ref{centralized_nomix} is applied to the first subsystem to recover the $\alpha F$ bits of each file, and the pipeline forward caching scheme is applied to the second subsystem to recover the $(1-\alpha)F$ bits of each file.

Consider the first subsystem. The equivalent file size, the user's cache memory and  the relay's memory are $\alpha F$, $\frac{M_1F}{\alpha F }=\frac{M_1}{\alpha}$, and $\frac{\beta M_2F}{\alpha F  }=\frac{\beta M_2}{\alpha}$, respectively.  By \eqref{eq_cen_nomix}, we obtain the transmission delay of  the first subsystem as 
\begin{IEEEeqnarray}{rCl}
R_\mathsf{s1,C}=R_\mathsf{p1,C}+R_\mathsf{p2,C}.
\end{IEEEeqnarray}
where $R_\mathsf{p1,C}$, $R_\mathsf{p2,C}$ are defined  in Theorem \ref{delay of central}.  

Similarly, consider the involved parameters in the second subsystem, the equivalent file size  and user cache memory are  $(1-\alpha)F$ and  $\frac{(1-\beta)M_2}{(1-\alpha)}$
, respectively,  thus the transmission delay of the second subsystem  by using pipe-line forward scheme is 
\begin{equation}
 R_\mathsf{s2,C}= r_\mathsf{c}\Big(\frac{(1-\beta)M_2}{(1-\alpha)N},K_1K_2\Big).
\end{equation}

Merge the two subsystems into a complete system, we obtain the transmission delay mentioned in Theorem  \ref{delay of central}:
\begin{equation}
T_\mathsf{Pro,C}=\alpha R_\mathsf{s1,C}+(1-\alpha)R_\mathsf{s2,C}.
\end{equation}

\section{Proof of Theorem \ref{theorem_rateDec}}\label{sec_proof_decentralized}
In this section, we first present a decentralized coded caching scheme, which is carefully designed to 
reduce the transmission delay for a two-layer network by caching files in both relays' cache memories and users' cache memories. And we combine this proposed scheme with pipeline-forward scheme, getting the mixed scheme with transmission delay shown in Theorem \ref{theorem_rateDec}.


\subsection{Concurrent Caching Scheme}\label{sec_CCS}

We first propose a caching scheme and carefully design the algorithms to allocate the delivery of subfiles  stored in the server and relays. The decentralized placement procedure is applied to cache memory in relays and users independently, which results in that subfiles may be stored only in users, only in relays or both in users and relays. For each user $u^i_j\in\mathcal{U}$, we divide the subfiles required to be delivered into three parts: Subfiles \uppercase\expandafter{\romannumeral 1} are  those cached by other relays except relay $i$, and will be sent using a block decode-forward strategy;   Subfiles \uppercase\expandafter{\romannumeral 2} are those not cached by any relay, and will be send by  the  pipeline-forward  strategy introduced in Example \ref{eg_pipe}; Subfiles \uppercase\expandafter{\romannumeral 3} are those cached by relay $i$, and will be sent by the single-layer decentralized caching scheme.  When the server sends  Subfiles \uppercase\expandafter{\romannumeral 2}, there are some parts which are redundant for relay $i$. We let  relay $i$ send parts of Subfiles III  if the server's signal is not useful. 





The scheme is divided into the placement phase and the delivery phase. The placement phase is summarized in Algorithm 2. Specifically, each relay $i$ uses  caching function $\phi _i$ to map the $N$ files into its $M_1F$-bit cache randomly and independently. Similarly, each {user} $u^i_j$ uses caching function $\phi _{j}^{i}$ to map the files into its $M_2F$-bit cache randomly and independently, which is showed in Algorithm 1. Therefore, each file $W_n$ is divided into multiple subfiles, i.e., for $n=1,\ldots,N$
\[W_n=\left\{W_{{n}, \mathcal{S}}^{\mathcal{Q}},~\textnormal{for all}~{\mathcal{Q}\subseteq[K_1],\mathcal{S}\subseteq\mathcal{U}}\right\}.\]
The subfiles stored in user $u^i_j$ is  \[\left\{W_{{1}, \mathcal{S}}^{\mathcal{Q}},\ldots,W_{{N}, \mathcal{S}}^{\mathcal{Q}},~\textnormal{for all}~{\mathcal{Q}\subseteq[K_1],\mathcal{S}\subseteq\mathcal{U}~\textnormal{with $u_j^i\in \mathcal{S}$}}\right\}.\]
The subfiles stored in relay $i$ is  \[\left\{W_{{1}, \mathcal{S}}^{\mathcal{Q}},\ldots,W_{{N}, \mathcal{S}}^{\mathcal{Q}},~\textnormal{for all}~{\mathcal{Q}\subseteq[K_1],\mathcal{S}\subseteq\mathcal{U}~\textnormal{with $i\in \mathcal{Q}$}}\right\}.\]

 Note that the request vector $\mathbf{d}$ is not informed during this phase and all caching functions select contents to cache completely arbitrarily. When the file size $F$ is large, by the law of large numbers, the subfile size with high probability can be written as
\begin{equation}
\begin{aligned}
\left|W_{n,\mathcal{S}}^{\mathcal{Q}}\right|\approx& \Big(\frac{M_1}{N}\Big)^{\left|\mathcal{Q}\right|}\cdot\Big(1-\frac{M_1}{N}\Big)^{K_1-\left|\mathcal{Q}\right|}\\&\cdot\Big(\frac{M_2}{N}\Big)^{\left|\mathcal{S}\right|}\cdot\Big(1-\frac{M_2}{N}\Big)^{K_1K_2-\left|\mathcal{S}\right|}F.
\end{aligned}
\end{equation}
It's easy to verified that under the placement given in Algorithm 1, each relay and user fill to $M_1F$ and $M_2F$ bits to its caches, respectively.

 \begin{algorithm}
\caption{Placement Phase}
\begin{algorithmic}[1]
\FOR{$n\in[N]$}
\FOR{$i\in[K_1]$}
\STATE{relay $i$ independently caches a subset of $\frac{M_1F}{N}$ bits of file $W_n$, chosen uniformly at random}
\FOR{$j\in[K_2]$}
\STATE{$\text{User}~u^i_j$ independently caches a subset of $\frac{M_2F}{N}$ bits of file $W_n$, chosen uniformly at random}
\ENDFOR
\ENDFOR
\ENDFOR
\end{algorithmic}
\end{algorithm}

In the delivery phase, each {user} $u^i_j$ requests a file  $W_{d_j^{~\!\!i}}$, and  the request vector $\mathbf{d}$ is promoted to the server and relays. Our objective is to get the upper bound on transmission delay in the worst request case, so we assume that each of the users makes unique request in the following discussion. 
The subfiles  of $W_{d_j^{~\!\!i}}$ requested by user $u^i_j$ can be characterized into three types:
\begin{itemize}
\item Subfiles \uppercase\expandafter{\romannumeral 1}: the subfiles cached by other relays except relay $i$, i.e., 
 $W_{d_j^{~\!\!i},\mathcal{S}\backslash \{u^i_j\}}^{\mathcal{Q}}$, for all $\mathcal{S}\subset\mathcal{U}, \mathcal{Q}\subset [K_1]$ with $i\notin\mathcal{Q}$.

\item Subfiles \uppercase\expandafter{\romannumeral 2}: the subfiles not cached by any relay, i.e.,  $W_{d_j^{~\!\!i},\mathcal{S}\backslash \{u^i_j\}}^{\emptyset}$, for all $\mathcal{S}\subset\mathcal{U}$.
\item Subfiles \uppercase\expandafter{\romannumeral 3}: the subfilles cached by relay $i$, i.e.,  $W_{d_j^{~\!\!i},\mathcal{S}\backslash \{u^i_j\}}^{\mathcal{Q}}$, for all $\mathcal{S}\subset\mathcal{U}, \mathcal{Q}\subseteq[K_1]$ with $i\in\mathcal{Q}$.
\end{itemize} 


Next we illustrate the transmission of the three types of subfiles. For ease of notation, the corresponding transmission  is labelled as Transmission \uppercase\expandafter{\romannumeral 1}, \uppercase\expandafter{\romannumeral 2}, \uppercase\expandafter{\romannumeral 3}, respectively.

\begin{itemize}
\item Transmission \uppercase\expandafter{\romannumeral 1}: Since there is no communication or cooperation among  the relays,   Subfiles \uppercase\expandafter{\romannumeral 1} can only be delivered from the server to users.     The server sends the following symbol to the relays
\begin{subequations}\label{eq_symbolS1}
\begin{IEEEeqnarray}{rCl}
&&\oplus_{i\in\mathcal{R}}\left(\oplus_{u^i_j\in\mathcal{S}}W_{d_j^{~\!\!i},\mathcal{S}_i\backslash \{u^i_j\}\cup{\mathcal{T}}}^{\mathcal{R}\backslash \{i\}}\right),
\end{IEEEeqnarray}
for each $j\in{[K_2]}$ and 
\begin{IEEEeqnarray}{rCl}
\mathcal{R}&\subseteq&[K_1]: |\mathcal{R}|=r, r=K_1,K_1-1,\ldots,2,\\
\mathcal{S}_i&\subseteq &\mathcal{U}_i:u^i_j\in\mathcal{S}_i~\text{and}~|\mathcal{S}_i|=s, s\in[K_2],\\
\mathcal{T}&\subseteq& \mathcal{U}\backslash \mathcal{U}_i:|\mathcal{T}|=t,t\in \{K_1K_2-K_2,\ldots,0\}.\quad 
\end{IEEEeqnarray}
\end{subequations}
For a given tuple of parameters $(j,\mathcal{R},\mathcal{S}_i,\mathcal{T})$, after   relay $i\in\mathcal{R}$ observing the symbol in \eqref{eq_symbolS1}, 
it  decodes  the  message   
\begin{IEEEeqnarray}{rCl}\label{eq_symbolS1Relayi}
\oplus_{u^i_j\in\mathcal{S}_i} W_{d_j^{~\!\!i},\mathcal{S}_i\backslash \{u^i_j\}\cup{\mathcal{T}}}^{\mathcal{R}\backslash \{i\}},
\end{IEEEeqnarray} and  forwards it to its attached users. 

After observing the message \eqref{eq_symbolS1Relayi}, each user $u^i_j\in\mathcal{S}_i$ decodes the following requested subfile based on the cached content \[W_{d_j^{~\!\!i},\mathcal{S}_i\backslash \{u^i_j\}\cup{\mathcal{T}}}^{\mathcal{R}\backslash \{i\}}
.\]

Notice that the relays can simultaneously receive and transmit signals, that means  when each relay   $i\in\mathcal{R}$ decodes and forwards the  message \eqref{eq_symbolS1Relayi},   the server can keep sending the symbol  for a different tuple of parameters $(j,\mathcal{R},\mathcal{S}_i,\mathcal{T})$. This procedure is similar to the block-Markov coding scheme in \cite{Cover'79}, where in every current block the transmitter sends a new source message, and the relay decodes and forwards the signal received from the previous block.

According to the delivery strategy described above,  Subfiles  \uppercase\expandafter{\romannumeral 1} can be perfectly known at  the corresponding users.  The delay of Transmission I, denoted by $\bar{R}_\mathsf{p1,D}$, is thus\footnote{The transmission in fact needs  some extra delay for sending the initial symbol, since the relays at the beginning do not have message $\eqref{eq_symbolS1Relayi}$, but this extra delay can be ignored if the number of symbols is sufficiently large, or by letting the relays send parts of Subfile III to the users without waiting the server's transmission. }
\begin{IEEEeqnarray}{rCl}\label{eq_rateT1}
\bar{R}_\mathsf{p1,D}&=&\Big[r_{\mathsf{d}}\Big(\frac{M_1}{N},K_1\big)\!-\!K_1\big(1\!-\!\frac{M_1}{N}\big)^{\!K_1}\Big] r_{\mathsf{d}}\big(\frac{M_2}{N},K_2\big).\nonumber\\ 
\end{IEEEeqnarray}


\item Transmission \uppercase\expandafter{\romannumeral 2}: Subfiles II also need to be sent originally from the the server. The server sends symbol
\begin{IEEEeqnarray}{rCl}\label{eq_subfileIIserver}
\oplus_{u_j^{~\!\!i}\in\mathcal{S}}W_{d_j^{~\!\!i},\mathcal{S}\backslash \{u_j^{~\!\!i}\}}^{\emptyset},
\end{IEEEeqnarray}
for each subset $\mathcal{S}\subseteq \mathcal{U}:|\mathcal{S}|=s, s\in{[K_1K_2]}$.
Thus the delay of sending all Subfiles II, denoted by $R_\textnormal{T2}$, is
\begin{IEEEeqnarray}{rCl}\label{eq_rateT2}
\bar{R}_\mathsf{p2,D}=\Big(1-\frac{M_1}{N}\Big)^{K_1} r_{\mathsf{d}}\Big(\frac{M_2}{N},K_1K_2\Big).
\end{IEEEeqnarray}

For a given subset $S$ with $S\! ~\cap~\!\mathcal{U}_i\neq\emptyset$, the symbol \eqref{eq_subfileIIserver} contains the requested subfiles for relay $i$'s attached users. We call this kind of symbol  the \emph{useful symbol} of Subfiles II for relay $i$. Relay $i$ uses pipe-line forward scheme as described in Example \ref{eg_pipe} to send the message   \[\oplus_{u_j^{~\!\!i}\in\mathcal{S}}W_{d_j^{~\!\!i},\mathcal{S}\backslash \{u_j^{~\!\!i}\}}^{\emptyset}.\]  
Each user $u^i_j\in\mathcal{S}$ decodes its requested subfile $W_{d_j^{~\!\!i},\mathcal{S}\backslash \{u_j^{~\!\!i}\}}^{\emptyset}$ based on the cached content.


For a given set subset $S$  with $S\! ~\cap~\!\mathcal{U}_i=\emptyset$,  the symbol \eqref{eq_subfileIIserver} does not contain any subfile requested by  relay $i$'s attached users. We call this kind of symbol  the \emph{redundant symbol} of Subfiles II for  relay $i$. The rate of this part, denoted by $\bar{R}_e$, can be computed as
\begin{IEEEeqnarray*}{rCl}
\bar{R}_\mathsf{e,D}=\Big(1-\frac{M_1}{N}\Big)^{\!K_1} \!\!\Big(1-\frac{M_2}{N}\Big)^{\!K_2} \!r_{\mathsf{d}}\Big(\frac{M_2}{N},(K_1\!-\!1)K_2\Big).
\end{IEEEeqnarray*}

Here, we exploit this spared time resource by letting each relay $i$ send some parts of Subfiles III. This is possible since  Subfiles III have already been stored in the relay's  cache memory during the placement phase.  More specifically, when the server sends the symbol \eqref{eq_subfileIIserver} with $S\! ~\cap~\!\mathcal{U}_i\neq\emptyset$, each relay $i$  sends
\begin{subequations}\label{eq_subfileIIIPart}
\begin{IEEEeqnarray}{rCl} \label{eq_subfileIIISymbol}
\oplus_{u_j^{~\!\!i}\in\mathcal{S}'_i}W_{d_j^{~\!\!i},\mathcal{S}'_i\backslash \{u_j^{~\!\!i}\}\cup{\mathcal{T}'}}^{\mathcal{R}'}
\end{IEEEeqnarray} 
for some set 
\begin{IEEEeqnarray}{rCl} \label{eq_subfileIIISet}
&&\mathcal{R}' \subseteq[K_1]: i\in \mathcal{R}',\\
&&\mathcal{S}'_i\subseteq \mathcal{U}_i:|\mathcal{S}_i|=s,s\in[K_2],\\
&&\mathcal{T}'\subseteq\mathcal{U} \backslash \mathcal{U}_i:|\mathcal{T}'|=t, t = K_1K_2-K_2, \ldots, 0,\quad 
\end{IEEEeqnarray} 
\end{subequations}
such that the rate of delivering these symbols equals to $\min\{\bar{R}_\mathsf{e,D},\bar{R}_\mathsf{p3,D}\}$, where $\bar{R}_\mathsf{p3,D}$ denotes the rate required to  send by relay nodes. Due to the different sizes of symbols in \eqref{eq_subfileIIserver} and \eqref{eq_subfileIIIPart}, we may not be able to find $(\mathcal{R}', \mathcal{S}'_i, \mathcal{T}'_i)$ such that the delay of sending symbols \eqref{eq_subfileIIIPart} exactly equals to $R_e$. One can obviate this problem by splitting  Subfiles III into the smaller pico-files and sending these pico-files in the same way as  \eqref{eq_subfileIIIPart}.

Each User $u^i_j$ decodes Subfiles II $W_{d_j^{~\!\!i},\mathcal{S}\backslash \{u_j^{~\!\!i}\}}^{\emptyset}$ and part of Subfiles III $W_{d_j^{~\!\!i},\mathcal{S}'_i\backslash \{u_j^{~\!\!i}\}\cup{\mathcal{T}'}}^{\mathcal{R}'}$ based on its cached content.

\item Transmission \uppercase\expandafter{\romannumeral 3}: Review the three subfiles summarized above, when Transmission \uppercase\expandafter{\romannumeral 1} and \uppercase\expandafter{\romannumeral 2} finished, only the remaining parts of the Subfiles \uppercase\expandafter{\romannumeral 3} need to be transmitted from the  relays to their attached users. The server does not send any symbol in Transmission \uppercase\expandafter{\romannumeral 3}.  All the relays concurrently send signals to their attached users, i.e.,  relay $i$ concurrently sends the symbol same as \eqref{eq_subfileIIIPart} except that  $(\mathcal{R}', \mathcal{S}'_i, \mathcal{T}')$ are replaced by $(\mathcal{R}'', \mathcal{S}''_i, \mathcal{T}'')$, respectively, and $\mathcal{R}'\cup \mathcal{R}''=[K_1]$, $\mathcal{S}'_i\cup \mathcal{S}''_i=\mathcal{U}_i$ and $\mathcal{T}'\cup \mathcal{T}''=\mathcal{U}\backslash\mathcal{U}$.

The rate for sending Subfiles III, denoted by $\bar{R}_\mathsf{p3,D}$, is 
\begin{IEEEeqnarray}{rCl}
\bar{R}_\mathsf{p3,D}=\frac{M_1}{N} r_{\mathsf{d}}\Big(\frac{M_2}{N},K_2\Big).
\end{IEEEeqnarray}
Thus the delay of Transmission III, denoted by $R_\mathsf{T_3,D}$,  is
\begin{IEEEeqnarray}{rCl}\label{eq_rateT3}
R_\mathsf{T_3,D}=\max(\bar{R}_\mathsf{p3,D}-\bar{R}_\mathsf{e,D},0).
\end{IEEEeqnarray}
Combines \eqref{eq_rateT1}, \eqref{eq_rateT2} and \eqref{eq_rateT3}, we obtain the  delay of the whole transmission, denoted by  $\bar{R}_\mathsf{s2,D}$, is
\begin{IEEEeqnarray}{rCl}\label{eq_delayPropose}
\bar{R}_\mathsf{s2,D}=\bar{R}_\mathsf{p1,D}+\bar{R}_\mathsf{p2,D}+ \max(\bar{R}_\mathsf{p3,D}-\bar{R}_\mathsf{e,D},0).\qquad
\end{IEEEeqnarray}

\end{itemize}

The delivery phase is summarized in Algorithm 3. Table \ref{Table_DeliveryExample} shows the   order of transmitting Subfiles I and II and III. 
In order to illustrate the caching and delivery scheme described above, we consider an example as below.

\begin{table}\centering
\caption{Delivery of Transmission I, II and III.}
\setlength{\tabcolsep}{0.5mm}{
\centering
\begin{tabular}{|c|c|c|c|}
\hline
Transmission &Server&Each Relay &Rate\\
\hline
I& \multicolumn{2}{c|}{Subfiles I} &  $\bar{R}_\mathsf{p1,D}$\\
\hline
\multirow{2}*{II} & \multirow{2}*{Subfiles II} & Useful symbols of Subfiles II & $\bar{R}_\mathsf{p2,D}-\bar{R}_\mathsf{e,D}$ \\ \cline{3-4}& &  Parts of Subfiles III& $\bar{R}_\mathsf{e,D}$\\
\hline
III & & Remaining parts of Subfiles III &$R_\mathsf{T_3,D}$\\
\hline
\end{tabular}}\label{Table_DeliveryExample}
\end{table}

\begin{algorithm}\label{algoDecentra}
\caption{Delivery Phase}
\begin{algorithmic}[1]
\STATE{\textbf{Procedure 1} DELIVERY of SUBFILES \uppercase\expandafter{\romannumeral 1}}
\FOR{$r=K_1,K_1-1,\ldots,2$}
		\FOR{$\mathcal{R}\subseteq [K_1]:|\mathcal{R}|=r$}
		\FOR{$s\in[K_2]$}
		\FOR{$j\in [K_2]$}
		\FOR{$i\in\mathcal{R}$}
		\FOR{$\mathcal{S}_i\subseteq \mathcal{U}_i:|\mathcal{S}_i|=s$  and $u^i_j\in\mathcal{S}_i$}
		\FOR{$t=K_1K_2-K_2,\ldots 1,0$}
		\FOR{$\mathcal{T}\subseteq \mathcal{U}\backslash \mathcal{U}_i:|\mathcal{T}|=t$}
		\STATE{$X_{\textbf{d}} \!\leftarrow \!\!\oplus_{i\in\mathcal{R}}\!\left(\!\oplus_{u^i_j\in\mathcal{S}_i}W_{d_j^{~\!\!i},\mathcal{S}_i\backslash \{u^i_j\}\cup{\mathcal{T}}}^{\mathcal{R}\backslash \{i\}}\right)$\\
		$X_{i,\textbf{d}}\leftarrow X_{\textbf{d}}$ }
		\ENDFOR
		\ENDFOR
		\ENDFOR
		\ENDFOR
		\ENDFOR
		\ENDFOR
		\ENDFOR
		\ENDFOR
\STATE{\textbf{End Procedure}}
\STATE{}
\STATE{\textbf{Procedure 2} DELIVERY of SUBFILES \uppercase\expandafter{\romannumeral 2}}
\FOR{$s \in [K_1K_2]$}
	\FOR{$\mathcal{S}\subseteq \mathcal{U}:|\mathcal{S}|=s$}
	\STATE{$X_{\textbf{d}}\leftarrow  \oplus_{u_j^{~\!\!i}\in\mathcal{S}}W_{d_j^{~\!\!i},\mathcal{S}\backslash \{u_j^{~\!\!i}\}}^{\emptyset}$\\
	 $X_{i,\textbf{d}}\leftarrow  \oplus_{u_j^{~\!\!i}\in\mathcal{S}}W_{d_j^{~\!\!i},\mathcal{S}\backslash \{u_j^{~\!\!i}\}}^{\emptyset}$ if $\mathcal{S}\cap \mathcal{U}_i\neq \emptyset$}
	\ENDFOR
	\ENDFOR
\STATE{\textbf{End Procedure}}
\STATE{}
\STATE{\textbf{Procedure 3} DELIVERY of SUBFILES  \uppercase\expandafter{\romannumeral 3}}
\STATE{Relay $i\in[K_1]$}
\FOR{$\mathcal{R} \subseteq[K_1]$ : $i\in \mathcal{R}$}
\FOR{$s \in [K_2]$}
	\FOR{$\mathcal{S}_i\subseteq \mathcal{U}_i:|\mathcal{S}_i|=s$}
	\FOR{$t = K_1K_2-K_2, \ldots, 0$}
	\FOR{$\mathcal{T}\subseteq\mathcal{U} \backslash \mathcal{U}_i:|\mathcal{T}|=t$}

	\STATE{$X_{i,\textbf{d}}\leftarrow \oplus_{u_j^{~\!\!i}\in\mathcal{S}_i}W_{d_j^{~\!\!i},\mathcal{S}_i\backslash \{u_j^{~\!\!i}\}\cup{\mathcal{T}}}^{\mathcal{R}}$}
	\ENDFOR
	\ENDFOR
	\ENDFOR
	\ENDFOR
	\ENDFOR
	\STATE{\textbf{End Procedure}}
	\end{algorithmic}
\end{algorithm}

\begin{Example}
Consider the two-layer network  with $N=4$ files ($A$, $B$, $C$ and $D$), $K_1=K_2=2$, and $M_1,M_2\in(0,N]$.

In the placement phase, using Algorithm 2, relay 1 and relay 2 independently store a random $M_1F/4$-bit subset of each file, and four users independently store a random $M_2F/4$-bit subset of each file.  Let $A_{\mathcal{S}}^{\mathcal{Q}}$ denotes the subfile of file $A$ that are stored in the cache memories of users in $\mathcal{S}$ and relays in $\mathcal{Q}$, where $\mathcal{S}\subset \{1, 2, 3, 4\}, \mathcal{Q}\subset \{1, 2\}$. For example, $A_{2,3}^{1}$ is the subfile of $A$ cached by user 2, 3 and relay 1. 


\begin{equation}
A=\left\{
\begin{aligned}
&A_{\emptyset}^{\emptyset},A_{1}^{\emptyset},A_{2}^{\emptyset},A_{3}^{\emptyset},\ldots,A_{1,2,3,4}^{\emptyset}\\
&A_{\emptyset}^{1},A_{1}^{1},A_{2}^{1},A_{3}^{1},\ldots,A_{1,2,3,4}^{1}\\
&A_{\emptyset}^{2},A_{1}^{2},A_{2}^{2},A_{3}^{2},\ldots,A_{1,2,3,4}^{2}\\
&A_{\emptyset}^{12},A_{1}^{1,2},A_{2}^{1,2},A_{3}^{12},\ldots,A_{1,2,3,4}^{1,2},\\
\end{aligned}
\right.
\end{equation}
\begin{equation}
\begin{aligned}
\left|W_{d_j^{~\!\!i},\ \mathcal{S}}^{\mathcal{Q}}\right|\approx& \Big(\frac{M_1}{4}\Big)^{\left|\mathcal{Q}\right|}\cdot\Big(1-\frac{M_1}{4}\Big)^{2-\left|\mathcal{Q}\right|}\\&\cdot\Big(\frac{M_2}{4}\Big)^{\left|\mathcal{S}\right|}\cdot\Big(1-\frac{M_2}{4}\Big)^{4-\left|\mathcal{S}\right|}F.
\end{aligned}
\end{equation}

In the delivery phase, we apply Algorithm 2 to send Subfiles \uppercase\expandafter{\romannumeral 1}, \uppercase\expandafter{\romannumeral 2} and \uppercase\expandafter{\romannumeral 3}.
More specifically,  the transmission for Subfiles \uppercase\expandafter{\romannumeral 1} is
\begin{equation*}
\begin{aligned}
&A_{\emptyset}^{2}\oplus C_{\emptyset}^{1},B_{\emptyset}^{2}\oplus D_{\emptyset}^{1},A_{2}^{2}\oplus B_{1}^{2}\oplus C_{4}^{1}\oplus D_{3}^{1},\\&A_{3}^{2}\oplus C_{1}^{1},A_{4}^{2}\oplus C_{2}^{1},B_{3}^{2}\oplus D_{1}^{1},B_{4}^{2}\oplus D_{2}^{1},\\&A_{2,3}^{2}\oplus B_{1,3}^{2}\oplus C_{1,4}^{1}\oplus D_{1,3}^{1},A_{2,4}^{2}\oplus B_{1,4}^{2}\oplus C_{2,4}^{1}\oplus D_{2,3}^{1}, \\&A_{3,4}^{2}\oplus C_{1,2}^{1},B_{3,4}^{2}\oplus D_{1,2}^{1},A_{2,3,4}^{2}\oplus B_{1,3,4}^{2}\oplus C_{1,2,4}^{1}\oplus D_{1,2,3}^{1}.
\end{aligned}
\end{equation*}
For a large enough file size $F$, this requires a communication rate
\begin{equation}\label{Rp1_ex}
\bar{R}_\mathsf{p1,D}=\frac{M_1}{4}\cdot\Big(1-\frac{M_1}{4}\Big)\cdot r_{\mathsf{d}}\Big(\frac{M_2}{4},2\Big).
\end{equation}
The transmission for Subfiles \uppercase\expandafter{\romannumeral 2} is
\begin{equation*}
\begin{aligned}
&A_{\emptyset}^{\emptyset},B_{\emptyset}^{\emptyset},C_{\emptyset}^{\emptyset},D_{\emptyset}^{\emptyset},\\&A_{2}^{\emptyset}\oplus B_{1}^{\emptyset},A_{3}^{\emptyset}\oplus C_{1}^{\emptyset},A_{4}^{\emptyset}\oplus D_{1}^{\emptyset},C_{2}^{\emptyset}\oplus B_{3}^{\emptyset},D_{2}^{\emptyset}\oplus B_{4}^{\emptyset},C_{4}^{\emptyset}\oplus D_{3}^{\emptyset},\\&A_{2,3}^{\emptyset}\oplus B_{1,3}^{\emptyset}\oplus C_{1,2}^{\emptyset},A_{2,4}^{\emptyset}\oplus B_{1,4}^{\emptyset}\oplus D_{1,2}^{\emptyset},\\&A_{3,4}^{\emptyset}\oplus C_{1,4}^{\emptyset}\oplus D_{1,3}^{\emptyset},B_{3,4}^{\emptyset}\oplus C_{2,4}^{\emptyset}\oplus D_{2,3}^{\emptyset},\\&A_{2,3,4}^{\emptyset}\oplus B_{1,3,4}^{\emptyset}\oplus C_{1,2,4}^{\emptyset}\oplus D_{1,2,3}^{\emptyset}.
\end{aligned}
\end{equation*}
The communication rate is
\begin{equation}\label{Rp2_ex}
\bar{R}_\mathsf{p2,D}=\Big(1-\frac{M_1}{4}\Big)^2\cdot r_{\mathsf{d}}\Big(\frac{M_2}{4},4\Big).
\end{equation}
The transmission for Subfiles \uppercase\expandafter{\romannumeral 3} contains two parts:
\begin{itemize}
\item From relay 1 to user 1 and 2
\begin{equation*}
\begin{aligned}
&A_{\emptyset}^{\mathcal{Q}_1},B_{\emptyset}^{\mathcal{Q}_1},A_{2}^{\mathcal{Q}_1}\oplus B_{1}^{\mathcal{Q}_1},A_{3}^{\mathcal{Q}_1},A_{4}^{\mathcal{Q}_1},B_{3}^{\mathcal{Q}_1},B_{4}^{\mathcal{Q}_1}\\&A_{2,3}^{\mathcal{Q}_1}\oplus B_{1,3}^{\mathcal{Q}_1},A_{2,4}^{\mathcal{Q}_1}\oplus B_{1,4}^{\mathcal{Q}_1},A_{3,4}^{\mathcal{Q}_1},B_{3,4}^{\mathcal{Q}_1},A_{2,3,4}^{\mathcal{Q}_1}\oplus B_{1,3,4}^{\mathcal{Q}_1}
\end{aligned}
\end{equation*}
\item From relay 2 to  user 3 and 4
\begin{equation*}
\begin{aligned}
&C_{\emptyset}^{\mathcal{Q}_2},D_{\emptyset}^{\mathcal{Q}_2},C_{4}^{\mathcal{Q}_2}\oplus D_{3}^{\mathcal{Q}_2},C_{1}^{\mathcal{Q}_2},C_{2}^{\mathcal{Q}_2},D_{1}^{\mathcal{Q}_2},D_{2}^{\mathcal{Q}_2}\\&C_{14}^{\mathcal{Q}_2}\oplus D_{1,3}^{\mathcal{Q}_2},C_{2,4}^{\mathcal{Q}_2}\oplus D_{2,3}^{\mathcal{Q}_2},C_{1,2}^{\mathcal{Q}_2},D_{1,2}^{\mathcal{Q}_2},C_{1,2,4}^{\mathcal{Q}_2}\oplus D_{1,2,3}^{\mathcal{Q}_2}
\end{aligned}
\end{equation*}
\end{itemize}
Here $\mathcal{Q}_1$ denotes the subsets of the relay $\{1,2\}$ which includes relay 1. And in this case, $\mathcal{Q}_1$ are $\{1\},\{1,2\}$. Similarly, $\mathcal{Q}_2$ denotes the subsets of the relay $\{1,2\}$ which includes relay 2, and here $\mathcal{Q}_2$ are $\{2\},\{1,2\}$.
Relay 1 and relay 2 can transmit these subfiles simultaneously and respectively, so the normalized rate is 
\begin{equation}\label{Rp3_ex}
\bar{R}_\mathsf{p3,D} = \frac{M_1}{4}\cdot r_{\mathsf{d}}\Big(\frac{M_2}{N},2\Big).
\end{equation}
The redundant symbols of  Subfiles II for the relays are as below:
\begin{itemize}
\item Relay 1 does not need
\begin{equation*}
C_{\emptyset}^{\emptyset},D_{\emptyset}^{\emptyset},C_{4}^{\emptyset}\oplus D_{3}^{\emptyset}.
\end{equation*}
\item Relay 2 does not need
\begin{equation*}
A_{\emptyset}^{\emptyset},B_{\emptyset}^{\emptyset},A_{2}^{\emptyset}\oplus B_{1}^{\emptyset}.
\end{equation*}
\end{itemize}
The rate of sending redundant symbols for each relay is
\begin{IEEEeqnarray}{rCl}\label{Re_ex}
\bar{R}_\mathsf{e,D}=\Big(1-\frac{M_1}{4}\Big)^2\Big(1-\frac{M_2}{4}\Big)^2\cdot r_{\mathsf{d}}\Big(\frac{M_2}{4},2\Big).
\end{IEEEeqnarray}

Thus, combine and \eqref{Rp1_ex}, \eqref{Rp2_ex}, \eqref{Rp3_ex} and \eqref{Re_ex} , the transmission delay  of this scheme is  
\begin{IEEEeqnarray}{rCl}
 \bar{R}_\mathsf{s1,D}\triangleq\bar{R}_\mathsf{p1,D}+\bar{R}_\mathsf{p2,D}+\max\{\bar{R}_\mathsf{p3,D}-\bar{R}_\mathsf{e,D},0\}.
\end{IEEEeqnarray}

\subsection{Hybrid Decentralized Scheme}\label{sec_CCS}
Now apply the similar method described in [3, Sec. \uppercase\expandafter{\romannumeral 5}-C] to combine the scheme described above with the pipeline-forward scheme. Denote the transmission delay caused by the pipeline-forward scheme as $R_\mathsf{s2,D}$, then we obtain the transmission delay\begin{equation}
T_\mathsf{Pro,D}=\alpha R_\mathsf{s1,D}+(1-\alpha)R_\mathsf{s2,D}.
\end{equation}


\end{Example}

\section{Conclusions}\label{sec_conclusion}
 
 In this paper, we proposed  coded caching schemes for the  cache-aided relay network, where a server accesses a library of files and wishes to communicate with users with the help of caches and relays.   We design a centralized and a decentralized  caching schemes that fully exploit the spared  time resource by allowing the concurrent transmission between the two layers. It is shown that both caching schemes are approximately optimal and can further reduce the transmission delay compared to  the previously known caching scheme.  Moreover, we  show that if each relay's caching size equals to a threshold, e.g., 38.2\% of full library's size for the two-relay case,  increasing the relay's caching size will not  reduce the transmission latency.



\appendices

\section{Choice of $\alpha$ and $\beta$}\label{pr_ab}
Recall $t_1=K_1M_1/N$ and $t_2=K_2M_2/N$.
\subsection{Choice of $(\alpha,\beta)$ for $T_\mathsf{Hcc,C}$ } 
Let
\begin{IEEEeqnarray*}{rCl}  
	\frac{\partial T_{\mathsf{Hcc,C}}}{\partial \beta}&=&\frac{-t_{2}\left(1+K_{2}\right)}{\left(1+\frac{\beta}{\alpha} t_{2}\right)^{2}}+\frac{K_{1} t_{2}\left(1+K_{1} K_{2}\right)}{\left(1+ \frac{K_{1}(1-\beta)}{\left(1-\alpha\right)} t_{2}\right)^{2}}=0,\nonumber\\
	 \stackrel{K_2\gg 1}{\Longrightarrow}\beta &=& \alpha+\frac{\alpha(1-\alpha)}{t_{2}}\left(\frac{1}{K_{1}}-1\right),\nonumber\\
 \stackrel{t_2\gg1}{\Longrightarrow} \beta &\approx &\alpha.
\end{IEEEeqnarray*}
Now we drive the optimal $\alpha$ when $\beta = \alpha$.  Due to the existence of functions $\min\{{M_1}/{(\alpha N)},0\}$ and $[\cdot]^+$, $\frac{\partial T_{\mathsf{Hcc,C}}}{\partial \alpha}$   can not be computed directly. We thus  consider the problem in two cases: $\alpha<  M_1/{N}$ and $\alpha \geq M_1/{N}$. \\
\emph{1) Case} $\alpha < {M_{1}}/{N}$: 
\begin{IEEEeqnarray}{rCl}
	T_{\mathsf{Hcc,C}}&=&\frac{\alpha}{1+t_{2}} \left(K_{2}-t_{2}\right)+(1-\alpha) \frac{K_{1} K_{2}+K_{1} t_{2}}{1+K_{1}t_{2}}, \nonumber \\
	\frac{\partial T_{\mathsf{Hcc,C}}}{\partial \alpha}&=&\frac{K_{2}-t_{2}}{1+t_{2}}-\frac{K_{2}+t_{2}}{\frac{1}{K_{1}}+t_{2}}<0.
\end{IEEEeqnarray}
Thus, in this case  $T_{\mathsf{Hcc,C}}$ is monotonically decreasing  and the  optimal $\alpha$ is ${M_{1}}/{N}$. \\
\emph{2) Case} $\alpha \geq {M_{1}}/{N}$: 
\begin{IEEEeqnarray}{rCl}
T_{\mathsf{Hcc,C}}& =& \frac{\alpha  K_{2}(K_1\!-\! t_{1})}{1\!+\!{t_{1}}/{\alpha}}+\alpha \frac{K_{2}\!-\!t_{2}}{1+t_{2}}\!+\!(1\!-\!\alpha) \frac{K_{1} (K_{2}- t_{2})}{1+K_{1} t_{2}}, \nonumber \quad \quad\\
\frac{\partial T_{\mathsf{Hcc,C}}}{\partial \alpha}&= &0\quad {\Longrightarrow}\nonumber\\
\alpha&=& t_{1}\frac{-  1+  \sqrt{1+\left(1-\frac{\Delta_1}{K_1K_2} \right)\left(\frac{1}{K_{2}}+\frac{\Delta_1}{K_1K_{2}}\right)}}{ \left(1-\frac{\Delta_1}{K_1K_2}\right)}\leq \frac{M_{1}}{N}\nonumber.
\end{IEEEeqnarray}
where $\Delta_1 \triangleq \frac{\left(K_{2}-t_{2}\right)\left(K_{1}-1\right)}{\left(1+K_{1} t_{2}\right)\left(1+t_{2}\right)}$. Thus, in this case    the  optimal $\alpha$ is ${M_{1}}/{N}$.

 From case 1 and 2, we conclude that $\beta = \alpha={M_{1}}/{N}$ is an approximately optimal choice.

\subsection{Choice of $(\alpha,\beta)$ for $T_\mathsf{Pro,C}$   
}
For easy of computation, we choose $\alpha = \beta$. Due to the existence of functions $\min\{{M_1}/{\alpha N},0\}$ and $[\cdot]^+$,   $\frac{\partial T_{\mathsf{Pro,C}}}{\partial \alpha}$ can not be computed directly. We thus   consider the problem  two cases: $\alpha<  M_1/{N}$ and $\alpha\geq M_1/{N}$. 
  \\
\emph{1) Case} $\alpha < {M_{1}}/{N}$: 
\begin{IEEEeqnarray*}{rCl}
T_{\mathsf{Pro,C}}& =& \frac{M_{1}t_{2}\left(\frac{1}{M_{2}}-1\right)}{1+t_{2}}+\frac{(1-\alpha)\left(K_{1} K_{2}+K_{1} t_{2}\right)}{1+K_{1} t_{2}},\\ 
\frac{\partial T_{\mathsf{Pro,C}}}{\partial \alpha}&= & -\frac{\left(K_{1} K_{2}+K_{1} t_{2}\right)}{1+K_{1} t_{2}}<0.
\end{IEEEeqnarray*}
Thus, the  optimal $\alpha$ can't be small than ${M_{1}}/{N}$. \\
\emph{2) Case} $\alpha \geq {M_{1}}/{N}$: 
\begin{IEEEeqnarray}{rCl}
T_{\mathsf{Pro,C}}  
		&=&\frac{\alpha}{1+t_{2}}\left[\frac{\left(K_{1}-\frac{t_{1}}{\alpha}\right)\left(K_{2}-t_{2}\right)}{1+\frac{t_{1}}{\alpha}}
		+\frac{M_{1}}{\alpha N}\left(K_{2}-t_{2}\right)\right] \nonumber \\
		&&+(1-\alpha) \frac{K_{1} K_{2}-K_{1} t_{2}}{1+K_{1} t_{2}},\nonumber\\
\frac{\partial T_{\mathsf{Pro,C}}}{\partial \alpha}&= &  0\quad {\Longrightarrow}\nonumber\\
 		\alpha&=&t_{1} \left(-1+\sqrt{1+(1+\Delta_2)/(K_{1}-\Delta_2)}\right),\nonumber\\
		\alpha &\approx& \gamma\triangleq t_{1} \frac{1+\Delta_2}{2\left(K_{1}-\Delta_2\right)}\geq \frac{M_1}{N} .
\end{IEEEeqnarray}
where $\Delta_2  \triangleq \frac{K_{1}\left(1+t_{2}\right)}{1+K_{1} t_{2}}\geq 1$. 
From case 1 and  2, and since $\alpha\in[0,1]$, we have 
\begin{IEEEeqnarray}{rCl} 
	\alpha^* = \beta^*  =\min\{\gamma,1\}.
\end{IEEEeqnarray}

\section{Proof of Theorem \ref{UselessM1}}\label{proof_TheoremUselssCen}

First consider the setup using centralized caching placement. By choosing $\alpha=\beta=1$, the achievable upper bounds  $T _\mathsf{Pro,C}$ in Theorem \ref{delay of central} reduces to
\begin{IEEEeqnarray}{rCl}
T_\mathsf{Pro,C}&=& \frac{M_{1}}{N}r_\mathsf{c}\left(\frac{M_{2}}{N},K_{2}\right)  +r_\mathsf{c}\left(\frac{M_{1}}{N}, K_{1}\right) r_\mathsf{c}\left(\frac{M_{2}}{N}, K_{2}\right).\nonumber \\
&=& r_\mathsf{c}\left(\frac{M_{2}}{N}, K_{2}\right) \left(\frac{M_{1}}{N}+r_\mathsf{c}\left(\frac{M_{1}}{N}, K_{1}\right)  \right).
\end{IEEEeqnarray}
By letting  
\[ \frac{M_{1}}{N}+r_\mathsf{c}\left(\frac{M_{1}}{N}, K_{1}\right)   =1,\]
we have $M_1=(K_1-1)N/K_1$. Thus, if $M_1$ is equal to the threshold 
$(K_1-1)N/K_1$, we can achieve $T_\mathsf{Pro,C}=r_\mathsf{c}\left(\frac{M_{2}}{N}, K_{2}\right)$, which is the same transmission delay as if each relay had the full library, i.e., $M_1=N$.

Now consider two-relay case $K_1=2$ using the decentralized caching placement. By choosing  $\alpha=\beta=1$, then $T_\mathsf{Pro,D}=R_\mathsf{s1,D}=R_\mathsf{p1,D}+R_\mathsf{p2,D}+\mathsf{max}\{R_\mathsf{p3,D}-R_\mathsf{e,D}, 0\}$.  If $R_\mathsf{p3,D}\geq R_\mathsf{e,D}$, i.e., 
\begin{equation}\label{eq:condition1}
\frac{M_1}{N}\ge \Big(1-\frac{M_1}{N}\Big)^2\Big(1-\frac{M_2}{N}\Big)^2,
\end{equation}
we have $T_\mathsf{Pro,D}= r_\mathsf{d}\left(\frac{M_{2}}{N}, K_{2}\right)$. Notice that condition $M_1\geq 0.38N$ ensure  \eqref{eq:condition1} to be satisfied, thus if $M_1= 0.38N$, we can achieve the same transmission delay as if each relay had the full library, i.e., $M_1=N$.



\section{Proof of Theorem \ref{theorem_compare}}\label{prTheo1}
\subsection{Proof of \eqref{eq_GapHCC}}


Although the optimal choice of our $T_\mathsf{Pro,D}$ could be different, we apply the same choice as HCC-III in \eqref{eqHCCab}, and compare $T_{\mathsf{Hcc,D}}^{\mathsf{Original}}$ with $T_{\mathsf{Pro,D}}$.
\subsubsection{Regime \uppercase\expandafter{\romannumeral 1}}
$M_1+M_2K_2\ge N$ and $0\le M_1\le N/4$. In this regime, scheme HCC-III chooses $(\alpha,\beta)=(M_1/N,M_1/N)$. we have
\begin{align*}
T_{\mathsf{Hcc,D}}^{\mathsf{Original}}&=r_\mathsf{d}\Big(\frac{M_2}{N},K_2\Big)+\Big(1-\frac{M_1}{N}\Big) r_\mathsf{d}\Big(\frac{M_2}{N},K_1K_2\Big),\\
T_\mathsf{Pro,D}&\leq\frac{M_1}{N} r_\mathsf{d}\Big(\frac{M_2}{N},K_2\Big)+\Big(1-\frac{M_1}{N}\Big) r_\mathsf{d}\Big(\frac{M_2}{N},K_1K_2\Big),
\end{align*}
and $T_{\mathsf{Hcc,D}}^{\mathsf{Original}} - T_\mathsf{Pro,D}\geq\Big(1-\frac{M_1}{N}\Big) r_\mathsf{d}\Big(\frac{M_2}{N},K_2\Big).$

\subsubsection{Regime \uppercase\expandafter{\romannumeral 2}}
$M_1+M_2K_2<N$. In this regime, scheme HCC-III chooses $(\alpha,\beta)=(M_1/(M_1+M_2K_2),0)$, we have
\begin{align*}
T_{\mathsf{Hcc,D}}^{\mathsf{Original}}=&\frac{M_1K_2}{M_1+M_2K_2} r_\mathsf{d}\Big(\frac{M_1+M_2K_2}{N},K_1\Big)\nonumber\\&+\frac{M_2K_2}{M_1+M_2K_2} r_\mathsf{d}\Big(\frac{M_1+M_2K_2}{NK_2},K_1K_2\Big)\nonumber\\&+\frac{M_2K_2}{M_1+M_2K_2} r_\mathsf{d}\Big(\frac{M_1+M_2K_2}{NK_2},K_2\Big)\nonumber\\&+\frac{M_1K_2}{M_1+M_2K_2},\\
T_\mathsf{Pro,D}\leq&\frac{M_1K_2}{M_1+M_2K_2} r_\mathsf{d}\Big(\frac{M_1+M_2K_2}{N},K_1\Big)+\frac{M_1K_2}{N}\nonumber\\&+\frac{M_2K_2}{M_1+M_2K_2} r_\mathsf{d}\Big(\frac{M_1+M_2K_2}{NK_2},K_1K_2\Big)\nonumber\\&-\frac{(K_1-1)M_1K_2}{M_1+M_2K_2} \Big(1-\frac{M_1+M_2K_2}{N}\Big)^{K_1},
\end{align*}
and
\begin{equation*}
\begin{aligned}
T_{\mathsf{Hcc,D}}^{\mathsf{Original}} - T_\mathsf{Pro,D}\geq &\frac{M_2K_2}{M_1+M_2K_2} r_\mathsf{d}\Big(\frac{M_1+M_2K_2}{NK_2},K_2\Big)\\&+ \frac{(K_1\!-\!1)M_1K_2}{M_1\!+\!M_2K_2} \Big(1\!-\!\frac{M_1\!+\!M_2K_2}{N}\Big)^{K_1}\\&+\Big(\frac{M_1K_2}{M_1+M_2K_2}-\frac{M_1K_2}{N}\Big)\\\overset{(a)}{\ge}&\frac{M_2K_2}{M_1+M_2K_2} r_\mathsf{d}\Big(\frac{M_1+M_2K_2}{NK_2},K_2\Big)
\end{aligned}
\end{equation*}•
where (a) follows from $M_1+M_2K_2<N$.

\subsubsection{Regime \uppercase\expandafter{\romannumeral 3}}
$M_1+M_2K_2\ge N$ and $N/4<M_1\le N$. In this regime, scheme HCC-III chooses $(\alpha,\beta)=(M_1/N,1/4)$, we have
\begin{align*}
T_{\mathsf{Hcc,D}}^{\mathsf{Original}} =&\Big(1-\frac{M_1}{N}\Big)  r_\mathsf{d}\Big(\frac{3M_2}{4(N-M_1)},K_1K_2\Big)\nonumber\\&+\Big(1-\frac{M_1}{N}\Big)  r_\mathsf{d}\Big(\frac{3M_2}{4(N-M_1)},K_2\Big)\nonumber\\&+\frac{M_1}{N}  r_\mathsf{d}\Big(\frac{M_2}{4M_1},K_2\Big),\\
T_\mathsf{Pro,D}\leq&\Big(1-\frac{M_1}{N}\Big)  r_\mathsf{d}\Big(\frac{3M_2}{4(N-M_1)},K_1K_2\Big)\nonumber\\&+\frac{M_1}{N}  r_\mathsf{d}\Big(\frac{M_2}{4M_1},K_2\Big),
\end{align*}
and $T_{\mathsf{Hcc,D}}^{\mathsf{Original}} - T_\mathsf{Pro,D}\geq \Big(1-\frac{M_1}{N}\Big) r_\mathsf{d}\Big(\frac{3M_2}{4(N-M_1)},K_2\Big).$

\subsection{Proof of $T_{\mathsf{Hcc}, \mathsf{C}}\geq T_{\mathsf{Pro}, \mathsf{C}}$}
From (\ref{modifed_HCC}) and (\ref{eqDelayCentral}), rewrite $T_{\mathsf{Hcc}, \mathsf{C}}$ and $T_{\mathsf{Pro}, \mathsf{C}}$, 
\begin{IEEEeqnarray*}{rCl}
T_{\mathsf{Hcc}, \mathsf{C}} &=& \underbrace{\alpha K_{2} r_{\mathsf{c}}\left(\frac{M_{1}}{\alpha N}, K_{2}\right)}_{\mathsf{A_1}}+\underbrace{\alpha r_{\mathsf{c}}\left(\frac{\beta M_{2}}{\alpha N}, K_{1}\right)}_\mathsf{B_1} \nonumber \\
&&\quad+\underbrace{(1-\alpha) r_{\mathsf{c}}\left(\frac{(1-\beta) M_{2}}{(1-\alpha) N}, K_{1} K_{2}\right)}_\mathsf{C_1} ,\nonumber\\
T_{\mathsf{Pro}, \mathsf{C}} &=& \frac{\mathsf{A_1}}{K_2}r_{\mathsf{c}}\left(\frac{\beta M_{2}}{\alpha N}, K_{2}\right) + \min \left\{\frac{M_{1}}{\alpha N}, 1\right\} \cdot \mathsf{B_1}+\mathsf{C_1}\nonumber\\
&\leq & \mathsf{A_1}+\mathsf{B_1}+\mathsf{C_1}=T_\mathsf{Hcc,C}.
\end{IEEEeqnarray*}
where the last equality holds by ${K_2}\geq r_{\mathsf{c}}\left({\beta M_{2}}/{(\alpha N)}, K_{2}\right)$.

\subsection{Proof of $T_{\mathsf{Hcc}, \mathsf{D}}\geq T_{\mathsf{Pro}, \mathsf{D}}$}
Rewrite $T_{\mathsf{Hcc}, \mathsf{D}}$ and $T_{\mathsf{Pro}, \mathsf{D}}$ as in  (\ref{modifed_HCC}) and (\ref{eq_proposeRate}), 
\begin{IEEEeqnarray}{rCl}\label{rateHCCDAppen}
T_{\mathsf{Hcc}, \mathsf{D}} &=& \underbrace{\alpha K_{2} r_{\mathsf{d}}\left(\frac{M_{1}}{\alpha N}, K_{1}\right)}_{\mathsf{A_2}}+\underbrace{\alpha r_{\mathsf{d}}\left(\frac{\beta M_{2}}{\alpha N}, K_{2}\right)}_\mathsf{B_2} \nonumber \\
&&\quad +\underbrace{(1-\alpha) r_{\mathsf{d}}\left(\frac{(1-\beta) M_{2}}{(1-\alpha) N}, K_{1} K_{2}\right)}_\mathsf{C_2} ,
\\
T_{\mathsf{Pro}, \mathsf{D}} &=& \alpha(R_\mathsf{p1,D}+R_\mathsf{p2,D}+\mathsf{max}\{R_\mathsf{p3,D}-R_\mathsf{e,D}, 0\})+\mathsf{C_2} 
\nonumber\\
&\leq & \alpha(R_\mathsf{p1,D}+R_\mathsf{p2,D}+ R_\mathsf{p3,D}-R_\mathsf{e,D})+\mathsf{C_2} 
\nonumber\\ 
&\leq & \alpha(R_\mathsf{p1,D}+R_\mathsf{p2,D}+ R_\mathsf{p3,D})+\mathsf{C_2}
\nonumber\\ 
&=&   \min \left\{\frac{M_{1}}{\alpha N}, 1\right\} \cdot \mathsf{B_2}+\mathsf{C_2}+\alpha r_{\mathsf{d}}\left(\frac{\beta M_{2}}{\alpha N}, K_{2}\right)\nonumber\\
&&\quad \cdot  \left[r_{\mathsf{d}}\Big(\frac{M_1}{\alpha N},K_1\Big)\!-\!K_1\Big(1\!-\!\frac{M_1}{ \alpha N}\Big)^{\!K_1}\right]
\nonumber\\
&&\quad +\alpha\Big(1-\frac{M_1}{\alpha N}\Big)^{K_1} r_{\mathsf{d}}\Big(\frac{\beta M_2}{\alpha N},K_1K_2\Big),
\nonumber\\ 
&  \stackrel{(a)}\leq & \alpha r_{\mathsf{d}}\Big(\frac{M_1}{\alpha N},K_1\Big)r_{\mathsf{d}}\left(\frac{\beta M_{2}}{\alpha N}, K_{2}\right)+\mathsf{B_2}+\mathsf{C_2}
\nonumber\\
& {\leq} & \mathsf{A_2}+ \mathsf{B_2} +\mathsf{C_2}=T_\mathsf{Hcc,D}.\nonumber
\end{IEEEeqnarray}
where (a) holds by $K_1r_{\mathsf{d}}\big(\frac{\beta M_2}{\alpha N},K_2\big) \geq r_{\mathsf{d}}\big(\frac{\beta M_2}{\alpha N},K_1K_2\big)$ and the last equality holds  by ${K_2}\geq r_{\mathsf{d}}\left({\beta M_{2}}/{(\alpha N)}, K_{2}\right)$.

\subsection{Proof of  $T_\mathsf{Hcc,C} \leq c_{1} T^*,$ and $T_{\mathsf{Hcc}, \mathsf{D}}\leq  c_2T^*$}
We first prove $T_\mathsf{Hcc,C} \leq c_{1} T^*,$ and $T_{\mathsf{Hcc}, \mathsf{D}}\leq  c_2T^*$ and then showed that $T_\mathsf{Pro,C}$ is within a smaller constant multiplicative gap than HCC scheme in \cite{Karamchandani'16}.

From \eqref{eq1_lowerbound},  rewrite the lower bound, for all $M_1,M_2\in[0,N]$, $s_1\in\{1,\ldots,K_1\}$ and $s, s_2\in\{1,\ldots,K_2\}$,
\begin{IEEEeqnarray}{rCl} \label{eq_lower2}
T^*\geq \max \{T_1^*,T_2^*\}
\end{IEEEeqnarray}
where 
\begin{IEEEeqnarray*}{rCl}
T^*_1&\triangleq&s_1s_2 -\frac{s_1M_1+s_1s_2M_2}{\lfloor N/(s_1s_2) \rfloor},\\
T_2^*&\triangleq&s-\frac{sM_2}{\lfloor N/s \rfloor}.
\end{IEEEeqnarray*}
In  \cite{Karamchandani'16} the authors consider the following three regimes: 
\begin{itemize}
\item regime I: $M_1+M_2K_2\ge N , 0\le M_1\le \frac{N}{4}$;
\item regime II: $M_1+M_2K_2<N$;
\item regime III: $M_1+M_2K_2\ge N,\frac{N}{4}\!<\!M_1\le N$,
\end{itemize}
and show that 
  $\mathsf{A}_2,\mathsf{B}_2,\mathsf{C}_2$ defined in \eqref{rateHCCDAppen} satisfies
\begin{IEEEeqnarray*}{rCl}
\mathsf{B}_2&\leq & c' T_2^*,\\
\mathsf{A}_2+\mathsf{C}_2&\leq&  c^{''} T_1^*,
\end{IEEEeqnarray*}
for some finite positive constants  $c'$ and $c^{''}$. With this result, it's easy to show that there must exist a finite positive constant $c_1$ such that \[T_\mathsf{Hcc,D}=\mathsf{A}_2+\mathsf{B}_2+\mathsf{C}_2\leq c_1T^*.\] Since $T_\mathsf{Pro,C}\leq T_\mathsf{Hcc,C}\leq T_\mathsf{Hcc,D}$, $T_\mathsf{Pro,D}\leq T_\mathsf{Hcc,D}$, we directly have
\begin{IEEEeqnarray}{rCl}
 T_\mathsf{Pro,C}\leq T_\mathsf{Hcc,C} \leq c_{1}\cdot T^*, \label{eqCompareC}\\  
   T_\mathsf{Pro,D}\leq T_\mathsf{Hcc,D} \leq c_{2}\cdot T^* \label{eqCompareD},
\end{IEEEeqnarray}
for some finite positive constants  $c_1$ and $c_2$.

In \cite{Karamchandani'16}, the authors showed that in regime  $M_1+M_2K_2\ge N$, when using  their HCC scheme, the rate of the first layer has a constant multiplicative gap of 35 within $T^*$. Now we show that our upper bound is within a smaller gap within $T^*$  in this regime. 

For  regime $M_1+M_2K_2\ge N$, choose $\alpha=\beta=1$, then from \eqref{eqDelayCentral} we have
\begin{IEEEeqnarray*}{rCl}
T_\mathsf{Pro,C} &=& r_\mathsf{c}\left(\frac{M_{2}}{ N},  K_{2}\right) \left(
\frac{M_1}{N}+\frac{1-M_1/N}{{1}/{K_1}+K_2M_2/N}\right)\nonumber\\
&\stackrel{(a)}{\leq} & 12\cdot T_2^*\cdot \left(\frac{M_1}{N}+\frac{K_2M_2/N}{{1}/{K_1}+K_2M_2/N}\right)\nonumber\\
&\leq& 12\cdot T_2^* \left(\frac{M_1}{N}+1 \right)\nonumber\\
&\leq & 24 \cdot T_2^*\leq 24\cdot T^*
\end{IEEEeqnarray*}
where (a) holds by condition $M_1+M_2K_2\ge N$ and  by  $r_\mathsf{c}\left({M_{2}}/{ N},  K_{2}\right)\leq 12\cdot T_2^*$, see proof in \cite{Centralized}.


\end{document}